\documentclass{llncs}
\usepackage[a4paper, margin=1.2in]{geometry}

\usepackage[pdftex,
        pdfauthor={Chun Tian},
        pdftitle={Further Formalization of the Process Algebra CCS in HOL4},
        pdfsubject={CCS},
        pdfkeywords={HOL},
        pdfstartview=FitH,
        bookmarksopen=true]{hyperref}

\usepackage{graphicx}
\usepackage{bussproofs}
\usepackage{amsmath}
\usepackage{amssymb}
\usepackage{mathrsfs}
\usepackage{multicol}
\usepackage{fancyvrb}
\usepackage[utf8]{inputenc}
\usepackage[T1]{fontenc}

\usepackage{bookmark}
\usepackage{listings}
\usepackage{holtexbasic}
\usepackage{alltt}
\newcommand{\HOLTokenObsCongr}{$\approx^c$}
\newcommand{\HOLTokenEPS}{$\overset{\epsilon}{\Rightarrow}$}
\newcommand{\HOLTokenSupC}{${}^c$}

\usepackage[all,cmtip]{xy}

\begin{document}

\title{Further Formalization of the Process Algebra CCS in HOL4}
\author{Chun Tian}
\institute{Scuola di Scienze, Universit\`{a} di Bologna\\
\email{chun.tian@studio.unibo.it}\\
Numero di matricola: 0000735539}
\maketitle

\begin{abstract}
In this project, we have extended previous work on the formalization
of the process algebra CCS in HOL4. We have added full supports on weak
bisimulation equivalence and observation congruence (rooted weak equivalence), with
related definitions, theorems and algebraic laws.
Some deep lemmas were also formally proved in this
project, including Deng Lemma, Hennessy Lemma and several versions of the
``Coarsest congruence contained in weak equivalence''. For the last
theorem, we have proved the full version (without any assumption) based on ordinals.
\end{abstract}

\section{Introduction}

The current project is a further extension of a previous project
\cite{Tian:2017vwba} on the formalization of the process algebra CCS
in HOL Theorem Prover (HOL4). In the previous work, we have successfully
covered the (strong) transitions of CCS processes, strong bisimulation
equivalence and have formally proved all the strong algebraic laws, including the expansion
law. But this is not a complete work for the formalization of CCS, as
in most model
checking cases, the specification and implementation of the same model
has only (rooted) weak equivalence. Thus a further extension to
previous work seems meaningful.

This project is still based on old Hol88 proof scripts written by Monica
Nesi, but it's not simply a porting of the remain old scripts
without new creations. Instead, we have essentially modified many
fundamental definitions and have proved many new theorems. And
with the changed definitions based on HOL4's rich theory library, previous
unprovable theorems now are provable. There're totally 200 theorems
and definitions in the project, now about 100 of them were newly
defined and proved by the author.

Below is a brief summary of changes and new features comparing with the old work:
\begin{enumerate}
\item We have extended the datatype definitions of CCS processes and
  transition actions, replacing all strings into general type
  variables. \footnote{This is not new invention, Prof. Nesi has done
    a similar change in her formalization of Value-passing CCS
     in 1999 \cite{Nesi:2017wo}. But work seems also done in Hol88,
    and related code is not available on Internet.} As a result, now it's
  possible to do reasoning on processes with limited number of actions
  and constants. In academic, the notation $CCS(a, b)$ represents the CCS
  subcalculus which can use at most $h$ constants 
and $k$ actions, and some important results hold for only certain CCS
subcalculus (e.g. \cite{gorriericcs}). With the new CCS datatype, now our
formalization has the ability to reason on 
this kind of CCS subclasulus. For almost all theorems and algebraic
laws, such a change has no affects, the only
exception is the ``coarsest congruence contained in $\approx$`` (Theorem
4.5 in \cite{Gorrieri:2015jt} or Proposition 3 in Chapter 7 of
\cite{Milner:2017tw}), in which the assumption is not automatically
true if the set of labels were finite. (We present two formal proofs of
this theorem with different assumptions in details in this project.)

\item We have completely turned to use HOL4's built-in supports of
  coinductive relations (\texttt{Hol_coreln}) for defining strong and weak bisimulation
  equivalences. As a result, many intermediate definitions and theorems towards the
  proof of \emph{Property (*) of strong and weak equivalence}
  \footnote{HOL theorem names: \texttt{PROPERTY_STAR} and \texttt{OBS_PROPERTY_STAR}} are
  not needed thus removed from the project.

\item We have extensively used HOL4's existing relationTheory and the
  supports of RTC (reflexive transitive closure) for defining the weak
  transitions of CCS processes. As a result, a large amount of cases
  and induction theorems were automatically available. It will show
  that, without these extra theorems (especially the right induction
  theorem) it's impossible to prove the transitivity of observation
  congruence.

\item In addition, we also formally proved the Hennessy Lemma and Deng Lemma (the weak
  equivalence version), which shows deep relations between weak equivalence and observation
  congruence. We used Deng Lemma to prove the hard part of Hennessy
  Lemma, therefore minimized the related proof scripts. These results
  have demonstrated that the author now has ability to convert
most informal proofs in CCS into formal proofs.

\item Some important theorems were not proved in the old Hol88 work,
  notable ones include:
 1) the transitivity of observation congruence
 (\texttt{OBS_CONGR_TRANS}) and 2) the
  coarsest congruence containing in weak equivalence (\texttt{PROP3}). In this project,
  we have finished these proofs with new related lemmas stated and
  proved.
 As for the `` coarsest congruence containing in weak
  equivalence'' theorem, we have successfully proved a stronger
  version (without any assumption) for finite-state CCS based on a new
  proof \cite{vanGlabbeek:2005ur} published by J. R. van Glabbeek in
  2005. Then we use HOL's ordinal theory and an axiomatized support of
  infinite sums of CCS to prove the full version for general CCS processes.

\item A rather complete theory of congruence on CCS has been newly
  built in this project. We have used
  this theory to explain the meaning of `` coarsest congruence containing in weak
  equivalence'' theorem.
\end{enumerate}

With these new additions, now the whole project has established the
comprehensive theory for pure CCS with all major results included. The
project can either be used as the theoretical basis for discovering new
results about CCS, or further developed into a model checking software
running in theorem prover.

This project is now part of official HOL4
repository\footnote{\url{https://github.com/HOL-Theorem-Prover/HOL}},
  under the directory \texttt{examples/CCS}.

\section{Background}

Our work in this and previous project were based on an old CCS formalization
\cite{Nesi:1992ve} built in Hol88 theorem prover (ancestry of HOL4), by
Monica Nesi during 1992-1995. As noted already in Background section
of the report of our previous project, The related proof scripts
mentioned in the publications of Prof. Nesi is not available on
Internet, but on June 7, 2016, Professor Nesi sent some old proof
scripts to the author in private, soon after the author asked for
these scripts in HOL mailing list. But these scripts did not include
any formalization for weak equivalence, rooted weak equivalence and
other things (e.g. HML) mentioned in her paper. At the beginning the
author thought that the rest scripts must have been lost, but it
turned out that this is not true.

On May 15, 2017, almost immediately after the author announced the
finish of the previous project to all related people, Prof. Nesi
replied the mail with the following contents:

\begin{quote}
``Dear Chun Tian,

Thanks a lot for your message, I am happy you were successful in your work! 
I will try and read your report as soon as I can, but in the meantime I would like 
to point out that my files on weak bisimulation, weak equivalence, observation 
congruence, modal logic, etc., are not lost.
In a mail to you (dated Jun 7th, 2016) I just sent you the first bunch of  files to 
start with. You said you were still learning HOL, so I thought it better not to 
"flood" you with all my files.
I don't know what your plans are now, but I would be glad to send you
other files on CCS in HOL if you fancy going on with this work.

Best regards,
Monica''
\end{quote}

Then it became obvious that, another further project on this topic should
be done in scope of the ``tirocinio'' (training) project \footnote{It is an obligatory part in
the author's study plan of Master degree in Computer Science} under the
supervision of Prof. Roberto Gorrieri in University of Bologna. And
instead of creating everything from
scratch, the author has another bunch of old scripts to start with. This is
a great advantage for doing another successful project.. After
having expressed such willings to Prof. Nesi, finally on
June 6, 2017, the author has received all the rest old proof script on
the formalization of pure CCS, covering weak bisimulation, weak
equivalence, observation congruence and HML. There're totally about
4000 lines of Classic ML code.

The old formalization of Hennessy-Milner Logic (for CCS) is not ported
into HOL4 in this project, because our focus in current project is mainly at the theorem proving
aspects, i.e. the proof of some deep theorems related to weak bisimulation equivalence
and observation congruence (rooted weak bisimulation
equivalence).  Actually, we have put aside one of the initial project
goals, i.e. creating a new model checking tool running in HOL theorem
prover. Instead, we have focused on pure theorem-proving staff in this
project, and deeply researched the current proofs for several important
theorems and the precise requirements to make these theorems hold.

\section{Extended CCS datatypes}

The type of CCS processes has been extended with two type variables:
$\alpha$ and $\beta$. $\alpha$ denotes the type of constants, and
$\beta$ denotes the type of labels. In HOL, such a higher order type is represented
as ``\HOLinline{(\ensuremath{\alpha}, \ensuremath{\beta}) \HOLTyOp{CCS}}''.  Whenever both type variables were
instantiated as \HOLinline{\HOLTyOp{string}}, the resulting type ``\HOLinline{(\HOLTyOp{string}, \HOLTyOp{string}) \HOLTyOp{CCS}}'' is equivalent with the \texttt{CCS} datatype in previous project.
Within the new settings, to represent CCS subcalculus like
$CCS(25,12)$, custom datatypes with limited number of
instances can be defined by users\footnote{In HOL, there's already a
  single-instance type \HOLinline{\HOLTyOp{unit}}, and a two-valued type
  \HOLinline{\HOLTyOp{bool}}, further custom datatypes can be defined by
  \texttt{Define} command.}. However we didn't go further in this
direction.

The type of transition labels were extended by type variable $\beta$,
the resulting new type is ``\HOLinline{\ensuremath{\beta} \HOLTyOp{Label}}'' in HOL. It's important
to notice that, for each possible value \HOLinline{\HOLFreeVar{l}} of the type
\HOLinline{\ensuremath{\beta}}, both ``\HOLinline{\HOLConst{name} \HOLFreeVar{l}}'' and ``\HOLinline{\HOLConst{coname} \HOLFreeVar{l}}'' are valid
labels, therefore the totally available number of labels are doubled
with the cardinality of the set of all possible values of type
\HOLinline{\ensuremath{\beta}}. Also noticed that, the
invisible action $\tau$ is part of the type ``\HOLinline{\ensuremath{\beta} \HOLTyOp{Action}}'', which
contains both $\tau$ and ``\HOLinline{\ensuremath{\beta} \HOLTyOp{Label}}'' values (wrapped by
constructor \HOLinline{\HOLConst{label}}. Thus if we count the number of all
possible \emph{actions} of the type ``\HOLinline{(\ensuremath{\alpha}, \ensuremath{\beta}) \HOLTyOp{CCS}}'', it
should be the doubled cardinality of type \HOLinline{\ensuremath{\alpha}} plus one.

Finally, it should be noticed that, in HOL, each valid type must
contain at least one value, thus in the minimal setting, there're
still three valid actions: $\tau$, the singleton input action and
output action. Whenever a CCS related theorem requires that ``there's
at least one non-$tau$ action'', such a requirement
can be omited from the assumptions of the theorem, because it's
automatically satisfied.

\section{Weak transitions and the EPS relation}

In previous project \cite{Tian:2017vwba}, we have discussed the advantage to use
\HOLinline{\HOLConst{EPS}} and \HOLinline{\HOLConst{WEAK_TRANS}} (instead of \HOLinline{\HOLConst{WEAK_TRACE}} used
in \cite{Gorrieri:2015jt}) for defining weak bisimulation, weak
bisimulation equivalence and observation congruence.  But we didn't
prove any theorem about \HOLinline{\HOLConst{EPS}} and \HOLinline{\HOLConst{WEAK_TRANS}} in previous
project. In this project, we have slightly changed the definition of
\HOLinline{\HOLConst{EPS}} with the helper definition \texttt{EPS0}
removed\footnote{After the author has learnt to use
  $\lambda$-expressions to express relations}:
\begin{definition}{(EPS)}
For any two CCS processes $E, E' \in Q$, define relation $EPS \subseteq
Q\times Q$ as the reflexive transitive closure (RTC) of single-$\tau$
transition between $E$ and $E'$ ($E \overset{\tau}{\longrightarrow} E'$):\footnote{In HOL4's
  \texttt{relationTheory}, the relation types is
  curried: instead of having the same type ``\HOLinline{\ensuremath{\alpha} \HOLTyOp{reln}}'' as the math definition, it has the type ``\HOLinline{\ensuremath{\alpha} -> \ensuremath{\alpha} -> \HOLTyOp{bool}}''. And the star(*) notation is for defining RTCs.}
\begin{alltt}
\HOLTokenTurnstile{} \HOLConst{EPS} \HOLSymConst{=} (\HOLTokenLambda{}\HOLBoundVar{E} \HOLBoundVar{E\sp{\prime}}. \HOLBoundVar{E} --\HOLSymConst{\ensuremath{\tau}}-> \HOLBoundVar{E\sp{\prime}})\HOLSymConst{\HOLTokenSupStar{}}
\end{alltt}
Intuitively speaking, \HOLinline{\HOLFreeVar{E} \HOLSymConst{\HOLTokenEPS} \HOLFreeVar{E\sp{\prime}}} (Math notion: $E
\overset{\epsilon}{\Longrightarrow} E'$) means there're zero or more $tau$-transitions from $p$ to $q$.
\end{definition}

Sometimes it's necessary to consider different transition cases when
\HOLinline{\HOLFreeVar{p} \HOLSymConst{\HOLTokenEPS} \HOLFreeVar{q}} holds, or induct on the number of $tau$ transitions
between $p$ and $q$. With such a
definition, beside the obvious reflexive and transitive
properties, a large amount of ``cases'' and induction theorem
already proved in HOL's \texttt{relationTheory} are immediately
available to us:
\begin{proposition}{(The ``cases'' theorem of the \HOLinline{\HOLConst{EPS}} relation)}
\begin{alltt}
\HOLTokenTurnstile{} \HOLFreeVar{x} \HOLSymConst{\HOLTokenEPS} \HOLFreeVar{y} \HOLSymConst{\HOLTokenEquiv{}} \HOLFreeVar{x} \HOLSymConst{=} \HOLFreeVar{y} \HOLSymConst{\HOLTokenDisj{}} \HOLSymConst{\HOLTokenExists{}}\HOLBoundVar{u}. \HOLFreeVar{x} --\HOLSymConst{\ensuremath{\tau}}-> \HOLBoundVar{u} \HOLSymConst{\HOLTokenConj{}} \HOLBoundVar{u} \HOLSymConst{\HOLTokenEPS} \HOLFreeVar{y} \hfill[EPS_cases1]
\HOLTokenTurnstile{} \HOLFreeVar{x} \HOLSymConst{\HOLTokenEPS} \HOLFreeVar{y} \HOLSymConst{\HOLTokenEquiv{}} \HOLFreeVar{x} \HOLSymConst{=} \HOLFreeVar{y} \HOLSymConst{\HOLTokenDisj{}} \HOLSymConst{\HOLTokenExists{}}\HOLBoundVar{u}. \HOLFreeVar{x} \HOLSymConst{\HOLTokenEPS} \HOLBoundVar{u} \HOLSymConst{\HOLTokenConj{}} \HOLBoundVar{u} --\HOLSymConst{\ensuremath{\tau}}-> \HOLFreeVar{y} \hfill[EPS_cases2]
\HOLTokenTurnstile{} \HOLFreeVar{E} \HOLSymConst{\HOLTokenEPS} \HOLFreeVar{E\sp{\prime}} \HOLSymConst{\HOLTokenEquiv{}} \HOLFreeVar{E} --\HOLSymConst{\ensuremath{\tau}}-> \HOLFreeVar{E\sp{\prime}} \HOLSymConst{\HOLTokenDisj{}} \HOLFreeVar{E} \HOLSymConst{=} \HOLFreeVar{E\sp{\prime}} \HOLSymConst{\HOLTokenDisj{}} \HOLSymConst{\HOLTokenExists{}}\HOLBoundVar{E\sb{\mathrm{1}}}. \HOLFreeVar{E} \HOLSymConst{\HOLTokenEPS} \HOLBoundVar{E\sb{\mathrm{1}}} \HOLSymConst{\HOLTokenConj{}} \HOLBoundVar{E\sb{\mathrm{1}}} \HOLSymConst{\HOLTokenEPS} \HOLFreeVar{E\sp{\prime}} \hfill[EPS_cases]
\end{alltt}
\end{proposition}

\begin{proposition}{(The induction and strong induction principles of
   the \HOLinline{\HOLConst{EPS}} relation)}
\begin{alltt}
\HOLTokenTurnstile{} (\HOLSymConst{\HOLTokenForall{}}\HOLBoundVar{x}. \HOLFreeVar{P} \HOLBoundVar{x} \HOLBoundVar{x}) \HOLSymConst{\HOLTokenConj{}} (\HOLSymConst{\HOLTokenForall{}}\HOLBoundVar{x} \HOLBoundVar{y} \HOLBoundVar{z}. \HOLBoundVar{x} --\HOLSymConst{\ensuremath{\tau}}-> \HOLBoundVar{y} \HOLSymConst{\HOLTokenConj{}} \HOLFreeVar{P} \HOLBoundVar{y} \HOLBoundVar{z} \HOLSymConst{\HOLTokenImp{}} \HOLFreeVar{P} \HOLBoundVar{x} \HOLBoundVar{z}) \HOLSymConst{\HOLTokenImp{}}
   \HOLSymConst{\HOLTokenForall{}}\HOLBoundVar{x} \HOLBoundVar{y}. \HOLBoundVar{x} \HOLSymConst{\HOLTokenEPS} \HOLBoundVar{y} \HOLSymConst{\HOLTokenImp{}} \HOLFreeVar{P} \HOLBoundVar{x} \HOLBoundVar{y} \hfill[EPS_ind]
\HOLTokenTurnstile{} (\HOLSymConst{\HOLTokenForall{}}\HOLBoundVar{x}. \HOLFreeVar{P} \HOLBoundVar{x} \HOLBoundVar{x}) \HOLSymConst{\HOLTokenConj{}} (\HOLSymConst{\HOLTokenForall{}}\HOLBoundVar{x} \HOLBoundVar{y} \HOLBoundVar{z}. \HOLBoundVar{x} --\HOLSymConst{\ensuremath{\tau}}-> \HOLBoundVar{y} \HOLSymConst{\HOLTokenConj{}} \HOLBoundVar{y} \HOLSymConst{\HOLTokenEPS} \HOLBoundVar{z} \HOLSymConst{\HOLTokenConj{}} \HOLFreeVar{P} \HOLBoundVar{y} \HOLBoundVar{z} \HOLSymConst{\HOLTokenImp{}} \HOLFreeVar{P} \HOLBoundVar{x} \HOLBoundVar{z}) \HOLSymConst{\HOLTokenImp{}}
   \HOLSymConst{\HOLTokenForall{}}\HOLBoundVar{x} \HOLBoundVar{y}. \HOLBoundVar{x} \HOLSymConst{\HOLTokenEPS} \HOLBoundVar{y} \HOLSymConst{\HOLTokenImp{}} \HOLFreeVar{P} \HOLBoundVar{x} \HOLBoundVar{y} \hfill[EPS_strongind]
\HOLTokenTurnstile{} (\HOLSymConst{\HOLTokenForall{}}\HOLBoundVar{x}. \HOLFreeVar{P} \HOLBoundVar{x} \HOLBoundVar{x}) \HOLSymConst{\HOLTokenConj{}} (\HOLSymConst{\HOLTokenForall{}}\HOLBoundVar{x} \HOLBoundVar{y} \HOLBoundVar{z}. \HOLFreeVar{P} \HOLBoundVar{x} \HOLBoundVar{y} \HOLSymConst{\HOLTokenConj{}} \HOLBoundVar{y} --\HOLSymConst{\ensuremath{\tau}}-> \HOLBoundVar{z} \HOLSymConst{\HOLTokenImp{}} \HOLFreeVar{P} \HOLBoundVar{x} \HOLBoundVar{z}) \HOLSymConst{\HOLTokenImp{}}
   \HOLSymConst{\HOLTokenForall{}}\HOLBoundVar{x} \HOLBoundVar{y}. \HOLBoundVar{x} \HOLSymConst{\HOLTokenEPS} \HOLBoundVar{y} \HOLSymConst{\HOLTokenImp{}} \HOLFreeVar{P} \HOLBoundVar{x} \HOLBoundVar{y} \hfill[EPS_ind_right]
\HOLTokenTurnstile{} (\HOLSymConst{\HOLTokenForall{}}\HOLBoundVar{x}. \HOLFreeVar{P} \HOLBoundVar{x} \HOLBoundVar{x}) \HOLSymConst{\HOLTokenConj{}} (\HOLSymConst{\HOLTokenForall{}}\HOLBoundVar{x} \HOLBoundVar{y} \HOLBoundVar{z}. \HOLFreeVar{P} \HOLBoundVar{x} \HOLBoundVar{y} \HOLSymConst{\HOLTokenConj{}} \HOLBoundVar{x} \HOLSymConst{\HOLTokenEPS} \HOLBoundVar{y} \HOLSymConst{\HOLTokenConj{}} \HOLBoundVar{y} --\HOLSymConst{\ensuremath{\tau}}-> \HOLBoundVar{z} \HOLSymConst{\HOLTokenImp{}} \HOLFreeVar{P} \HOLBoundVar{x} \HOLBoundVar{z}) \HOLSymConst{\HOLTokenImp{}}
   \HOLSymConst{\HOLTokenForall{}}\HOLBoundVar{x} \HOLBoundVar{y}. \HOLBoundVar{x} \HOLSymConst{\HOLTokenEPS} \HOLBoundVar{y} \HOLSymConst{\HOLTokenImp{}} \HOLFreeVar{P} \HOLBoundVar{x} \HOLBoundVar{y} \hfill[EPS_strongind_right]
\HOLTokenTurnstile{} (\HOLSymConst{\HOLTokenForall{}}\HOLBoundVar{E} \HOLBoundVar{E\sp{\prime}}. \HOLBoundVar{E} --\HOLSymConst{\ensuremath{\tau}}-> \HOLBoundVar{E\sp{\prime}} \HOLSymConst{\HOLTokenImp{}} \HOLFreeVar{P} \HOLBoundVar{E} \HOLBoundVar{E\sp{\prime}}) \HOLSymConst{\HOLTokenConj{}} (\HOLSymConst{\HOLTokenForall{}}\HOLBoundVar{E}. \HOLFreeVar{P} \HOLBoundVar{E} \HOLBoundVar{E}) \HOLSymConst{\HOLTokenConj{}}
   (\HOLSymConst{\HOLTokenForall{}}\HOLBoundVar{E} \HOLBoundVar{E\sb{\mathrm{1}}} \HOLBoundVar{E\sp{\prime}}. \HOLFreeVar{P} \HOLBoundVar{E} \HOLBoundVar{E\sb{\mathrm{1}}} \HOLSymConst{\HOLTokenConj{}} \HOLFreeVar{P} \HOLBoundVar{E\sb{\mathrm{1}}} \HOLBoundVar{E\sp{\prime}} \HOLSymConst{\HOLTokenImp{}} \HOLFreeVar{P} \HOLBoundVar{E} \HOLBoundVar{E\sp{\prime}}) \HOLSymConst{\HOLTokenImp{}}
   \HOLSymConst{\HOLTokenForall{}}\HOLBoundVar{x} \HOLBoundVar{y}. \HOLBoundVar{x} \HOLSymConst{\HOLTokenEPS} \HOLBoundVar{y} \HOLSymConst{\HOLTokenImp{}} \HOLFreeVar{P} \HOLBoundVar{x} \HOLBoundVar{y} \hfill[EPS_INDUCT]
\end{alltt}
\end{proposition}

Then we define the weak transition between two CCS processes upon the
\HOLinline{\HOLConst{EPS}} relation:
\begin{definition}
For any two CCS processes $E, E' \in Q$, define ``weak transition'' relation $\Longrightarrow \subseteq
Q\times A \times Q$, where A can be $\tau$ or a visible action:
$E \overset{a}{\longrightarrow} E'$ if and only if there exists two processes
$E_1$ and $E_2$ such that $E \overset{\epsilon}{\Longrightarrow} E_1
\overset{a}{\longrightarrow} E_2 \overset{\epsilon}{\Longrightarrow} E'$:
\begin{alltt}
\HOLTokenTurnstile{} \HOLFreeVar{E} ==\HOLFreeVar{u}=>> \HOLFreeVar{E\sp{\prime}} \HOLSymConst{\HOLTokenEquiv{}} \HOLSymConst{\HOLTokenExists{}}\HOLBoundVar{E\sb{\mathrm{1}}} \HOLBoundVar{E\sb{\mathrm{2}}}. \HOLFreeVar{E} \HOLSymConst{\HOLTokenEPS} \HOLBoundVar{E\sb{\mathrm{1}}} \HOLSymConst{\HOLTokenConj{}} \HOLBoundVar{E\sb{\mathrm{1}}} --\HOLFreeVar{u}-> \HOLBoundVar{E\sb{\mathrm{2}}} \HOLSymConst{\HOLTokenConj{}} \HOLBoundVar{E\sb{\mathrm{2}}} \HOLSymConst{\HOLTokenEPS} \HOLFreeVar{E\sp{\prime}}\hfill[WEAK_TRANS]
\end{alltt}
\end{definition}

Using above two definitions and the ``cases'' and induction theorems,
a large amount of properties about \HOLinline{\HOLConst{EPS}} and \HOLinline{\HOLConst{WEAK_TRANS}}
were proved:
\begin{proposition}{(Properties of \HOLinline{\HOLConst{EPS}} and \HOLinline{\HOLConst{WEAK_TRANS}})}
\begin{enumerate}
\item Any transition also implies a weak transition:
\begin{alltt}
\HOLTokenTurnstile{} \HOLFreeVar{E} --\HOLFreeVar{u}-> \HOLFreeVar{E\sp{\prime}} \HOLSymConst{\HOLTokenImp{}} \HOLFreeVar{E} ==\HOLFreeVar{u}=>> \HOLFreeVar{E\sp{\prime}}\hfill [TRANS_IMP_WEAK_TRANS]
\end{alltt}

\item Weak $\tau$-transition  implies \HOLinline{\HOLConst{EPS}} relation:
\begin{alltt}
\HOLTokenTurnstile{} \HOLFreeVar{E} ==\HOLSymConst{\ensuremath{\tau}}=>> \HOLFreeVar{E\sp{\prime}} \HOLSymConst{\HOLTokenImp{}} \HOLFreeVar{E} \HOLSymConst{\HOLTokenEPS} \HOLFreeVar{E\sp{\prime}}\hfill [WEAK_TRANS_TAU]
\end{alltt}

\item $\tau$-transition implies \HOLinline{\HOLConst{EPS}} relation:
\begin{alltt}
\HOLTokenTurnstile{} \HOLFreeVar{E} --\HOLSymConst{\ensuremath{\tau}}-> \HOLFreeVar{E\sp{\prime}} \HOLSymConst{\HOLTokenImp{}} \HOLFreeVar{E} \HOLSymConst{\HOLTokenEPS} \HOLFreeVar{E\sp{\prime}}\hfill [TRANS_TAU_IMP_EPS]
\end{alltt}

\item Weak $\tau$-transition implies an $\tau$ transition followed by
  EPS transition:
\begin{alltt}
\HOLTokenTurnstile{} \HOLFreeVar{E} ==\HOLSymConst{\ensuremath{\tau}}=>> \HOLFreeVar{E\sp{\prime}} \HOLSymConst{\HOLTokenImp{}} \HOLSymConst{\HOLTokenExists{}}\HOLBoundVar{E\sb{\mathrm{1}}}. \HOLFreeVar{E} --\HOLSymConst{\ensuremath{\tau}}-> \HOLBoundVar{E\sb{\mathrm{1}}} \HOLSymConst{\HOLTokenConj{}} \HOLBoundVar{E\sb{\mathrm{1}}} \HOLSymConst{\HOLTokenEPS} \HOLFreeVar{E\sp{\prime}}\hfill[WEAK_TRANS_TAU_IMP_TRANS_TAU]
\end{alltt}

\item \HOLinline{\HOLConst{EPS}} implies $\tau$-prefixed \HOLinline{\HOLConst{EPS}}:
\begin{alltt}
\HOLTokenTurnstile{} \HOLFreeVar{E} \HOLSymConst{\HOLTokenEPS} \HOLFreeVar{E\sp{\prime}} \HOLSymConst{\HOLTokenImp{}} \HOLSymConst{\ensuremath{\tau}}\HOLSymConst{..}\HOLFreeVar{E} \HOLSymConst{\HOLTokenEPS} \HOLFreeVar{E\sp{\prime}}\hfill [TAU_PREFIX_EPS]
\end{alltt}

\item Weak $\tau$-transition implies $\tau$-prefixed weak:
  $\tau$-transition:
\begin{alltt}
\HOLTokenTurnstile{} \HOLFreeVar{E} ==\HOLFreeVar{u}=>> \HOLFreeVar{E\sp{\prime}} \HOLSymConst{\HOLTokenImp{}} \HOLSymConst{\ensuremath{\tau}}\HOLSymConst{..}\HOLFreeVar{E} ==\HOLFreeVar{u}=>> \HOLFreeVar{E\sp{\prime}} \hfill [TAU_PREFIX_WEAK_TRANS]
\end{alltt}

\item A weak transition wrapped by EPS transitions is still a weak
  transition:
\begin{alltt}
\HOLTokenTurnstile{} \HOLFreeVar{E} \HOLSymConst{\HOLTokenEPS} \HOLFreeVar{E\sb{\mathrm{1}}} \HOLSymConst{\HOLTokenConj{}} \HOLFreeVar{E\sb{\mathrm{1}}} ==\HOLFreeVar{u}=>> \HOLFreeVar{E\sb{\mathrm{2}}} \HOLSymConst{\HOLTokenConj{}} \HOLFreeVar{E\sb{\mathrm{2}}} \HOLSymConst{\HOLTokenEPS} \HOLFreeVar{E\sp{\prime}} \HOLSymConst{\HOLTokenImp{}} \HOLFreeVar{E} ==\HOLFreeVar{u}=>> \HOLFreeVar{E\sp{\prime}}\hfill [EPS_AND_WEAK]
\end{alltt}

\item A weak transition after a $\tau$-transition is still a weak
  transition:
\begin{alltt}
\HOLTokenTurnstile{} \HOLFreeVar{E} --\HOLSymConst{\ensuremath{\tau}}-> \HOLFreeVar{E\sb{\mathrm{1}}} \HOLSymConst{\HOLTokenConj{}} \HOLFreeVar{E\sb{\mathrm{1}}} ==\HOLFreeVar{u}=>> \HOLFreeVar{E\sp{\prime}} \HOLSymConst{\HOLTokenImp{}} \HOLFreeVar{E} ==\HOLFreeVar{u}=>> \HOLFreeVar{E\sp{\prime}}\hfill [TRANS_TAU_AND_WEAK]
\end{alltt}

\item Any transition followed by an EPS transition becomes a weak
  transition:
\begin{alltt}
\HOLTokenTurnstile{} \HOLFreeVar{E} --\HOLFreeVar{u}-> \HOLFreeVar{E\sb{\mathrm{1}}} \HOLSymConst{\HOLTokenConj{}} \HOLFreeVar{E\sb{\mathrm{1}}} \HOLSymConst{\HOLTokenEPS} \HOLFreeVar{E\sp{\prime}} \HOLSymConst{\HOLTokenImp{}} \HOLFreeVar{E} ==\HOLFreeVar{u}=>> \HOLFreeVar{E\sp{\prime}}\hfill [TRANS_AND_EPS]
\end{alltt}

\item An EPS transition implies either no transition or a weak
  $\tau$-transition:
\begin{alltt}
\HOLTokenTurnstile{} \HOLFreeVar{E} \HOLSymConst{\HOLTokenEPS} \HOLFreeVar{E\sp{\prime}} \HOLSymConst{\HOLTokenImp{}} \HOLFreeVar{E} \HOLSymConst{=} \HOLFreeVar{E\sp{\prime}} \HOLSymConst{\HOLTokenDisj{}} \HOLFreeVar{E} ==\HOLSymConst{\ensuremath{\tau}}=>> \HOLFreeVar{E\sp{\prime}}\hfill [EPS_IMP_WEAK_TRANS]
\end{alltt}

\item Two possible cases for the first step of a weak transition:
\begin{alltt}
\HOLTokenTurnstile{} \HOLFreeVar{E} ==\HOLFreeVar{u}=>> \HOLFreeVar{E\sb{\mathrm{1}}} \HOLSymConst{\HOLTokenImp{}}
   (\HOLSymConst{\HOLTokenExists{}}\HOLBoundVar{E\sp{\prime}}. \HOLFreeVar{E} --\HOLSymConst{\ensuremath{\tau}}-> \HOLBoundVar{E\sp{\prime}} \HOLSymConst{\HOLTokenConj{}} \HOLBoundVar{E\sp{\prime}} ==\HOLFreeVar{u}=>> \HOLFreeVar{E\sb{\mathrm{1}}}) \HOLSymConst{\HOLTokenDisj{}}
   \HOLSymConst{\HOLTokenExists{}}\HOLBoundVar{E\sp{\prime}}. \HOLFreeVar{E} --\HOLFreeVar{u}-> \HOLBoundVar{E\sp{\prime}} \HOLSymConst{\HOLTokenConj{}} \HOLBoundVar{E\sp{\prime}} \HOLSymConst{\HOLTokenEPS} \HOLFreeVar{E\sb{\mathrm{1}}}\hfill [WEAK_TRANS_cases1]
\end{alltt}

\item The weak transition version of SOS inference rule $(Sum_1)$ and $(Sum_2)$:
\begin{alltt}
\HOLTokenTurnstile{} \HOLFreeVar{E} ==\HOLFreeVar{u}=>> \HOLFreeVar{E\sb{\mathrm{1}}} \HOLSymConst{\HOLTokenImp{}} \HOLFreeVar{E} \HOLSymConst{+} \HOLFreeVar{E\sp{\prime}} ==\HOLFreeVar{u}=>> \HOLFreeVar{E\sb{\mathrm{1}}}\hfill [WEAK_SUM1]
\HOLTokenTurnstile{} \HOLFreeVar{E} ==\HOLFreeVar{u}=>> \HOLFreeVar{E\sb{\mathrm{1}}} \HOLSymConst{\HOLTokenImp{}} \HOLFreeVar{E\sp{\prime}} \HOLSymConst{+} \HOLFreeVar{E} ==\HOLFreeVar{u}=>> \HOLFreeVar{E\sb{\mathrm{1}}}\hfill [WEAK_SUM2]
\end{alltt}
\end{enumerate}
\end{proposition}

\section{Weak bisimulation equivalence}

The concepts of weak bisimulation and weak bisimulation equivalence
(a.k.a. observation equivalence), together with the algebraic laws for
weak bisimulation equivalence, stand at a central position in this
project. This is mostly because all the deep theorems (Deng lemma,
Hennessy lemma, Coarsest congruence contained in weak equivalence) that we have
formally proved in this project, were all talking about the
relationship between weak bisimulation equivalence and rooted weak
bisimulation equivalence (a.k.a. observation congruence, we'll use
this shorted names in the rest of the paper).  The other reason is,
since the observation congruence is not recursively defined but rely
on the definition of weak equivalence, it turns out that, the
properties of weak equivalence were heavily used in the proof of
properties of observation congruence.

On the other side, it's quite easy to derive out almost all the
algebraic laws for weak equivalence (and observation congruence),
simply because strong equivalence implies weak equivalence (and also
observation congruence). This fact also reflects the fact that,
although strong equivalence and its algebraic laws were usually
useless in real world model checking, they do have contributions for
deriving more useful algebraic laws. And from the view of theorem
proving it totally make sense: if we try to prove any algebraic law
for weak equivalence \emph{directly}, the proof will be quite long and
difficult, and the handling of $tau$-transitions will be a common part
in all these proofs. But if we use the strong algebraic laws as
lemmas, the proofs were actually divided into two logical parts: one for
handling the algebraic law itself, the other for handling the $\tau$-transitions.

The definition of weak bisimulation is the same as in
\cite{Gorrieri:2015jt}, except for the use of EPS in case of
$\tau$-transitions:
\begin{definition}{(Weak bisimulation)}
\begin{alltt}
\HOLTokenTurnstile{} \HOLConst{WEAK_BISIM} \HOLFreeVar{Wbsm} \HOLSymConst{\HOLTokenEquiv{}}
   \HOLSymConst{\HOLTokenForall{}}\HOLBoundVar{E} \HOLBoundVar{E\sp{\prime}}.
     \HOLFreeVar{Wbsm} \HOLBoundVar{E} \HOLBoundVar{E\sp{\prime}} \HOLSymConst{\HOLTokenImp{}}
     (\HOLSymConst{\HOLTokenForall{}}\HOLBoundVar{l}.
        (\HOLSymConst{\HOLTokenForall{}}\HOLBoundVar{E\sb{\mathrm{1}}}.
           \HOLBoundVar{E} --\HOLConst{label} \HOLBoundVar{l}-> \HOLBoundVar{E\sb{\mathrm{1}}} \HOLSymConst{\HOLTokenImp{}}
           \HOLSymConst{\HOLTokenExists{}}\HOLBoundVar{E\sb{\mathrm{2}}}. \HOLBoundVar{E\sp{\prime}} ==\HOLConst{label} \HOLBoundVar{l}=>> \HOLBoundVar{E\sb{\mathrm{2}}} \HOLSymConst{\HOLTokenConj{}} \HOLFreeVar{Wbsm} \HOLBoundVar{E\sb{\mathrm{1}}} \HOLBoundVar{E\sb{\mathrm{2}}}) \HOLSymConst{\HOLTokenConj{}}
        \HOLSymConst{\HOLTokenForall{}}\HOLBoundVar{E\sb{\mathrm{2}}}.
          \HOLBoundVar{E\sp{\prime}} --\HOLConst{label} \HOLBoundVar{l}-> \HOLBoundVar{E\sb{\mathrm{2}}} \HOLSymConst{\HOLTokenImp{}}
          \HOLSymConst{\HOLTokenExists{}}\HOLBoundVar{E\sb{\mathrm{1}}}. \HOLBoundVar{E} ==\HOLConst{label} \HOLBoundVar{l}=>> \HOLBoundVar{E\sb{\mathrm{1}}} \HOLSymConst{\HOLTokenConj{}} \HOLFreeVar{Wbsm} \HOLBoundVar{E\sb{\mathrm{1}}} \HOLBoundVar{E\sb{\mathrm{2}}}) \HOLSymConst{\HOLTokenConj{}}
     (\HOLSymConst{\HOLTokenForall{}}\HOLBoundVar{E\sb{\mathrm{1}}}. \HOLBoundVar{E} --\HOLSymConst{\ensuremath{\tau}}-> \HOLBoundVar{E\sb{\mathrm{1}}} \HOLSymConst{\HOLTokenImp{}} \HOLSymConst{\HOLTokenExists{}}\HOLBoundVar{E\sb{\mathrm{2}}}. \HOLBoundVar{E\sp{\prime}} \HOLSymConst{\HOLTokenEPS} \HOLBoundVar{E\sb{\mathrm{2}}} \HOLSymConst{\HOLTokenConj{}} \HOLFreeVar{Wbsm} \HOLBoundVar{E\sb{\mathrm{1}}} \HOLBoundVar{E\sb{\mathrm{2}}}) \HOLSymConst{\HOLTokenConj{}}
     \HOLSymConst{\HOLTokenForall{}}\HOLBoundVar{E\sb{\mathrm{2}}}. \HOLBoundVar{E\sp{\prime}} --\HOLSymConst{\ensuremath{\tau}}-> \HOLBoundVar{E\sb{\mathrm{2}}} \HOLSymConst{\HOLTokenImp{}} \HOLSymConst{\HOLTokenExists{}}\HOLBoundVar{E\sb{\mathrm{1}}}. \HOLBoundVar{E} \HOLSymConst{\HOLTokenEPS} \HOLBoundVar{E\sb{\mathrm{1}}} \HOLSymConst{\HOLTokenConj{}} \HOLFreeVar{Wbsm} \HOLBoundVar{E\sb{\mathrm{1}}} \HOLBoundVar{E\sb{\mathrm{2}}}
\end{alltt}
\end{definition}

Weak bisimulation has some common properties:
\begin{proposition}{Properties of weak bisimulation}
\begin{enumerate}
\item The identity relation is a weak bisimulation:
\begin{alltt}
\HOLTokenTurnstile{} \HOLConst{WEAK_BISIM} (\HOLTokenLambda{}\HOLBoundVar{x} \HOLBoundVar{y}. \HOLBoundVar{x} \HOLSymConst{=} \HOLBoundVar{y})\hfill[IDENTITY_WEAK_BISIM]
\end{alltt}
\item The converse of a weak bisimulation is still a weak
  bisimulation:
\begin{alltt}
\HOLTokenTurnstile{} \HOLConst{WEAK_BISIM} \HOLFreeVar{Wbsm} \HOLSymConst{\HOLTokenImp{}} \HOLConst{WEAK_BISIM} (\HOLTokenLambda{}\HOLBoundVar{x} \HOLBoundVar{y}. \HOLFreeVar{Wbsm} \HOLBoundVar{y} \HOLBoundVar{x})\hfill[IDENTITY_WEAK_BISIM]
\end{alltt}
\item The composition of two weak bisimulations is a weak
  bisimulation:
\begin{alltt}
\HOLTokenTurnstile{} \HOLConst{WEAK_BISIM} \HOLFreeVar{Wbsm\sb{\mathrm{1}}} \HOLSymConst{\HOLTokenConj{}} \HOLConst{WEAK_BISIM} \HOLFreeVar{Wbsm\sb{\mathrm{2}}} \HOLSymConst{\HOLTokenImp{}}
   \HOLConst{WEAK_BISIM} (\HOLTokenLambda{}\HOLBoundVar{x} \HOLBoundVar{z}. \HOLSymConst{\HOLTokenExists{}}\HOLBoundVar{y}. \HOLFreeVar{Wbsm\sb{\mathrm{1}}} \HOLBoundVar{x} \HOLBoundVar{y} \HOLSymConst{\HOLTokenConj{}} \HOLFreeVar{Wbsm\sb{\mathrm{2}}} \HOLBoundVar{y} \HOLBoundVar{z})\hfill[COMP_WEAK_BISIM]
\end{alltt}
\item The union of two weak bisimulations is a weak bisimulation:
\begin{alltt}
\HOLTokenTurnstile{} \HOLConst{WEAK_BISIM} \HOLFreeVar{Wbsm\sb{\mathrm{1}}} \HOLSymConst{\HOLTokenConj{}} \HOLConst{WEAK_BISIM} \HOLFreeVar{Wbsm\sb{\mathrm{2}}} \HOLSymConst{\HOLTokenImp{}}
   \HOLConst{WEAK_BISIM} (\HOLTokenLambda{}\HOLBoundVar{x} \HOLBoundVar{y}. \HOLFreeVar{Wbsm\sb{\mathrm{1}}} \HOLBoundVar{x} \HOLBoundVar{y} \HOLSymConst{\HOLTokenDisj{}} \HOLFreeVar{Wbsm\sb{\mathrm{2}}} \HOLBoundVar{x} \HOLBoundVar{y})\hfill[UNION_WEAK_BISIM]
\end{alltt}
\end{enumerate}
\end{proposition}

There're two ways to define weak bisimulation equivalence in HOL4, one
is to define it as the union of all weak bisimulations:
\begin{definition}{(Alternative definition of weak equivalence)}
For any two CCS processes $E$ and $E'$, they're \emph{weak
  bisimulation equivalent} (or weak bisimilar) if and only if there's
a weak bisimulation relation between $E$ and $E'$:
\begin{alltt}
\HOLTokenTurnstile{} \HOLFreeVar{E} \HOLSymConst{\ensuremath{\approx}} \HOLFreeVar{E\sp{\prime}} \HOLSymConst{\HOLTokenEquiv{}} \HOLSymConst{\HOLTokenExists{}}\HOLBoundVar{Wbsm}. \HOLBoundVar{Wbsm} \HOLFreeVar{E} \HOLFreeVar{E\sp{\prime}} \HOLSymConst{\HOLTokenConj{}} \HOLConst{WEAK_BISIM} \HOLBoundVar{Wbsm}\hfill[WEAK_EQUIV]
\end{alltt}
\end{definition}
This is the old method used by Prof. Nesi in Hol88 in which there's no
support yet for defining co-inductive relations.  The new method we
have used in this project, is to use HOL4's new co-inductive relation
defining facility \texttt{Hol_coreln} to define weak bisimulation
equivalence:
\begin{lstlisting}
val (WEAK_EQUIV_rules, WEAK_EQUIV_coind, WEAK_EQUIV_cases) = Hol_coreln `
    (!(E :('a, 'b) CCS) (E' :('a, 'b) CCS).
       (!l.
	 (!E1. TRANS E  (label l) E1 ==>
	       (?E2. WEAK_TRANS E' (label l) E2 /\ WEAK_EQUIV E1 E2)) /\
	 (!E2. TRANS E' (label l) E2 ==>
	       (?E1. WEAK_TRANS E  (label l) E1 /\ WEAK_EQUIV E1 E2))) /\
       (!E1. TRANS E  tau E1 ==> (?E2. EPS E' E2 /\ WEAK_EQUIV E1 E2)) /\
       (!E2. TRANS E' tau E2 ==> (?E1. EPS E  E1 /\ WEAK_EQUIV E1 E2))
      ==> WEAK_EQUIV E E')`;
\end{lstlisting}
The disadvantage of this new method is that,  the rules used in above
definition actually duplicated the definition of weak bisimulation,
while the advantage is that, HOL4 automatically proved an important
theorem and returned it as the third return value of above
definition. This theorem is also called ``the property (*)'' (in
Milner's book \cite{Milner:2017tw}:
\begin{proposition}{(The property (*) for weak bisimulation
    equivalence)}
\begin{alltt}
\HOLTokenTurnstile{} \HOLFreeVar{a\sb{\mathrm{0}}} \HOLSymConst{\ensuremath{\approx}} \HOLFreeVar{a\sb{\mathrm{1}}} \HOLSymConst{\HOLTokenEquiv{}}
   (\HOLSymConst{\HOLTokenForall{}}\HOLBoundVar{l}.
      (\HOLSymConst{\HOLTokenForall{}}\HOLBoundVar{E\sb{\mathrm{1}}}.
         \HOLFreeVar{a\sb{\mathrm{0}}} --\HOLConst{label} \HOLBoundVar{l}-> \HOLBoundVar{E\sb{\mathrm{1}}} \HOLSymConst{\HOLTokenImp{}}
         \HOLSymConst{\HOLTokenExists{}}\HOLBoundVar{E\sb{\mathrm{2}}}. \HOLFreeVar{a\sb{\mathrm{1}}} ==\HOLConst{label} \HOLBoundVar{l}=>> \HOLBoundVar{E\sb{\mathrm{2}}} \HOLSymConst{\HOLTokenConj{}} \HOLBoundVar{E\sb{\mathrm{1}}} \HOLSymConst{\ensuremath{\approx}} \HOLBoundVar{E\sb{\mathrm{2}}}) \HOLSymConst{\HOLTokenConj{}}
      \HOLSymConst{\HOLTokenForall{}}\HOLBoundVar{E\sb{\mathrm{2}}}.
        \HOLFreeVar{a\sb{\mathrm{1}}} --\HOLConst{label} \HOLBoundVar{l}-> \HOLBoundVar{E\sb{\mathrm{2}}} \HOLSymConst{\HOLTokenImp{}}
        \HOLSymConst{\HOLTokenExists{}}\HOLBoundVar{E\sb{\mathrm{1}}}. \HOLFreeVar{a\sb{\mathrm{0}}} ==\HOLConst{label} \HOLBoundVar{l}=>> \HOLBoundVar{E\sb{\mathrm{1}}} \HOLSymConst{\HOLTokenConj{}} \HOLBoundVar{E\sb{\mathrm{1}}} \HOLSymConst{\ensuremath{\approx}} \HOLBoundVar{E\sb{\mathrm{2}}}) \HOLSymConst{\HOLTokenConj{}}
   (\HOLSymConst{\HOLTokenForall{}}\HOLBoundVar{E\sb{\mathrm{1}}}. \HOLFreeVar{a\sb{\mathrm{0}}} --\HOLSymConst{\ensuremath{\tau}}-> \HOLBoundVar{E\sb{\mathrm{1}}} \HOLSymConst{\HOLTokenImp{}} \HOLSymConst{\HOLTokenExists{}}\HOLBoundVar{E\sb{\mathrm{2}}}. \HOLFreeVar{a\sb{\mathrm{1}}} \HOLSymConst{\HOLTokenEPS} \HOLBoundVar{E\sb{\mathrm{2}}} \HOLSymConst{\HOLTokenConj{}} \HOLBoundVar{E\sb{\mathrm{1}}} \HOLSymConst{\ensuremath{\approx}} \HOLBoundVar{E\sb{\mathrm{2}}}) \HOLSymConst{\HOLTokenConj{}}
   \HOLSymConst{\HOLTokenForall{}}\HOLBoundVar{E\sb{\mathrm{2}}}. \HOLFreeVar{a\sb{\mathrm{1}}} --\HOLSymConst{\ensuremath{\tau}}-> \HOLBoundVar{E\sb{\mathrm{2}}} \HOLSymConst{\HOLTokenImp{}} \HOLSymConst{\HOLTokenExists{}}\HOLBoundVar{E\sb{\mathrm{1}}}. \HOLFreeVar{a\sb{\mathrm{0}}} \HOLSymConst{\HOLTokenEPS} \HOLBoundVar{E\sb{\mathrm{1}}} \HOLSymConst{\HOLTokenConj{}} \HOLBoundVar{E\sb{\mathrm{1}}} \HOLSymConst{\ensuremath{\approx}} \HOLBoundVar{E\sb{\mathrm{2}}}\hfill[OBS_PROPERTY_STAR]
\end{alltt}
\end{proposition}
It's known that, above property cannot be used as an alternative
definition of weak equivalence, because it doesn't capture all
possible weak equivalences. But it turns out that, for the proof of
most theorems about weak bisimularities this property is enough to be
used as a rewrite rule in their proofs. And, if we had used the old
method to define weak equivalence, it's quite difficult to prove above
property (*).\footnote{In our previous project, the property (*) for
  strong equivalence was proved based on the old method, then in this
  project we have completely removed these code and now both strong
  and weak bisimulation equivalences were based on the new method. On
  the other side, the fact that Prof. Nesi can define co-inductive relation
  without using \texttt{Hol_coreln} has shown that, the core HOL logic
doesn't need to be extended to suppport co-inductive relation, and all
what \texttt{Hol_coreln} does internally is to use the existing HOL
theorems to construct the related proofs. This is very different with
the situation in other theorem provers (e.g. Coq) in which the core
logic has to be extended to support co-induction.}

Using the alternative definition of weak equivalence, it's quite
simple to prove that, the weak equivalence is an equivalence relation:
\begin{proposition}{(Weak equivalence is an equivalence relation)}
\begin{alltt}
\HOLTokenTurnstile{} \HOLConst{equivalence} (\HOLSymConst{=~})\hfill[WEAK_EQUIV_equivalence]
\end{alltt}
or
\begin{alltt}
\HOLTokenTurnstile{} \HOLFreeVar{E} \HOLSymConst{\ensuremath{\approx}} \HOLFreeVar{E}\hfill[WEAK_EQUIV_REFL]
\HOLTokenTurnstile{} \HOLFreeVar{E} \HOLSymConst{\ensuremath{\approx}} \HOLFreeVar{E\sp{\prime}} \HOLSymConst{\HOLTokenImp{}} \HOLFreeVar{E\sp{\prime}} \HOLSymConst{\ensuremath{\approx}} \HOLFreeVar{E}\hfill[WEAK_EQUIV_SYM]
\HOLTokenTurnstile{} \HOLFreeVar{E} \HOLSymConst{\ensuremath{\approx}} \HOLFreeVar{E\sp{\prime}} \HOLSymConst{\HOLTokenConj{}} \HOLFreeVar{E\sp{\prime}} \HOLSymConst{\ensuremath{\approx}} \HOLFreeVar{E\sp{\prime\prime}} \HOLSymConst{\HOLTokenImp{}} \HOLFreeVar{E} \HOLSymConst{\ensuremath{\approx}} \HOLFreeVar{E\sp{\prime\prime}}\hfill[WEAK_EQUIV_TRANS]
\end{alltt}
\end{proposition}

The substitutability of weak equivalence under various CCS process
operators were then proved based on above definition and property
(*). However, as we know weak equivalence is not a congruence, in some
of these substitutability theorems we must added extra assumptions on
the processes involved, i.e. the stability of CCS processes:
\begin{definition}{(Stable processes (agents))}
A process (or agent) is said to be \emph{stable} if there's no $\tau$-transition
coming from it's root:
\begin{alltt}
\HOLTokenTurnstile{} \HOLConst{STABLE} \HOLFreeVar{E} \HOLSymConst{\HOLTokenEquiv{}} \HOLSymConst{\HOLTokenForall{}}\HOLBoundVar{u} \HOLBoundVar{E\sp{\prime}}. \HOLFreeVar{E} --\HOLBoundVar{u}-> \HOLBoundVar{E\sp{\prime}} \HOLSymConst{\HOLTokenImp{}} \HOLBoundVar{u} \HOLSymConst{\HOLTokenNotEqual{}} \HOLSymConst{\ensuremath{\tau}}
\end{alltt}
\end{definition}
Notice that, the stability of a CCS process doesn't imply the
$\tau$-free of all its sub-processes. Instead the definition only
concerns on the first transition leading from the process (root).

Among other small lemmas, we have proved the following properties of weak bisimulation
equivalence:
\begin{proposition}{Properties of weak bisimulation equivalence)}
\begin{enumerate}
\item Weak equivalence is substitutive under prefix operator:
\begin{alltt}
\HOLTokenTurnstile{} \HOLFreeVar{E} \HOLSymConst{\ensuremath{\approx}} \HOLFreeVar{E\sp{\prime}} \HOLSymConst{\HOLTokenImp{}} \HOLSymConst{\HOLTokenForall{}}\HOLBoundVar{u}. \HOLBoundVar{u}\HOLSymConst{..}\HOLFreeVar{E} \HOLSymConst{\ensuremath{\approx}} \HOLBoundVar{u}\HOLSymConst{..}\HOLFreeVar{E\sp{\prime}}\hfill[WEAK_EQUIV_SUBST_PREFIX]
\end{alltt}
\item Weak equivalence of stable agents is preserved by binary
  summation:
\begin{alltt}
\HOLTokenTurnstile{} \HOLFreeVar{E\sb{\mathrm{1}}} \HOLSymConst{\ensuremath{\approx}} \HOLFreeVar{E\sb{\mathrm{1}}\sp{\prime}} \HOLSymConst{\HOLTokenConj{}} \HOLConst{STABLE} \HOLFreeVar{E\sb{\mathrm{1}}} \HOLSymConst{\HOLTokenConj{}} \HOLConst{STABLE} \HOLFreeVar{E\sb{\mathrm{1}}\sp{\prime}} \HOLSymConst{\HOLTokenConj{}} \HOLFreeVar{E\sb{\mathrm{2}}} \HOLSymConst{\ensuremath{\approx}} \HOLFreeVar{E\sb{\mathrm{2}}\sp{\prime}} \HOLSymConst{\HOLTokenConj{}} \HOLConst{STABLE} \HOLFreeVar{E\sb{\mathrm{2}}} \HOLSymConst{\HOLTokenConj{}}
   \HOLConst{STABLE} \HOLFreeVar{E\sb{\mathrm{2}}\sp{\prime}} \HOLSymConst{\HOLTokenImp{}}
   \HOLFreeVar{E\sb{\mathrm{1}}} \HOLSymConst{+} \HOLFreeVar{E\sb{\mathrm{2}}} \HOLSymConst{\ensuremath{\approx}} \HOLFreeVar{E\sb{\mathrm{1}}\sp{\prime}} \HOLSymConst{+} \HOLFreeVar{E\sb{\mathrm{2}}\sp{\prime}}\hfill[WEAK_EQUIV_PRESD_BY_SUM]
\end{alltt}
\item Weak equivalence of stable agents is substitutive under binary
  summation on the right:
\begin{alltt}
\HOLTokenTurnstile{} \HOLFreeVar{E} \HOLSymConst{\ensuremath{\approx}} \HOLFreeVar{E\sp{\prime}} \HOLSymConst{\HOLTokenConj{}} \HOLConst{STABLE} \HOLFreeVar{E} \HOLSymConst{\HOLTokenConj{}} \HOLConst{STABLE} \HOLFreeVar{E\sp{\prime}} \HOLSymConst{\HOLTokenImp{}} \HOLSymConst{\HOLTokenForall{}}\HOLBoundVar{E\sp{\prime\prime}}. \HOLFreeVar{E} \HOLSymConst{+} \HOLBoundVar{E\sp{\prime\prime}} \HOLSymConst{\ensuremath{\approx}} \HOLFreeVar{E\sp{\prime}} \HOLSymConst{+} \HOLBoundVar{E\sp{\prime\prime}}
\end{alltt}
\item Weak equivalence of stable agents is substitutive under binary
  summation on the left:
\begin{alltt}
\HOLTokenTurnstile{} \HOLFreeVar{E} \HOLSymConst{\ensuremath{\approx}} \HOLFreeVar{E\sp{\prime}} \HOLSymConst{\HOLTokenConj{}} \HOLConst{STABLE} \HOLFreeVar{E} \HOLSymConst{\HOLTokenConj{}} \HOLConst{STABLE} \HOLFreeVar{E\sp{\prime}} \HOLSymConst{\HOLTokenImp{}} \HOLSymConst{\HOLTokenForall{}}\HOLBoundVar{E\sp{\prime\prime}}. \HOLBoundVar{E\sp{\prime\prime}} \HOLSymConst{+} \HOLFreeVar{E} \HOLSymConst{\ensuremath{\approx}} \HOLBoundVar{E\sp{\prime\prime}} \HOLSymConst{+} \HOLFreeVar{E\sp{\prime}}
\end{alltt}
\item Weak equivalence is preserved by parallel operator:
\begin{alltt}
\HOLTokenTurnstile{} \HOLFreeVar{E\sb{\mathrm{1}}} \HOLSymConst{\ensuremath{\approx}} \HOLFreeVar{E\sb{\mathrm{1}}\sp{\prime}} \HOLSymConst{\HOLTokenConj{}} \HOLFreeVar{E\sb{\mathrm{2}}} \HOLSymConst{\ensuremath{\approx}} \HOLFreeVar{E\sb{\mathrm{2}}\sp{\prime}} \HOLSymConst{\HOLTokenImp{}} \HOLFreeVar{E\sb{\mathrm{1}}} \HOLSymConst{||} \HOLFreeVar{E\sb{\mathrm{2}}} \HOLSymConst{\ensuremath{\approx}} \HOLFreeVar{E\sb{\mathrm{1}}\sp{\prime}} \HOLSymConst{||} \HOLFreeVar{E\sb{\mathrm{2}}\sp{\prime}}\hfill[WEAK_EQUIV_PRESD_BY_PAR]
\end{alltt}
\item Weak equivalence is substitutive under restriction operator:
\begin{alltt}
\HOLTokenTurnstile{} \HOLFreeVar{E} \HOLSymConst{\ensuremath{\approx}} \HOLFreeVar{E\sp{\prime}} \HOLSymConst{\HOLTokenImp{}} \HOLSymConst{\HOLTokenForall{}}\HOLBoundVar{L}. \HOLSymConst{\ensuremath{\nu}} \HOLBoundVar{L} \HOLFreeVar{E} \HOLSymConst{\ensuremath{\approx}} \HOLSymConst{\ensuremath{\nu}} \HOLBoundVar{L} \HOLFreeVar{E\sp{\prime}}\hfill[WEAK_EQUIV_SUBST_RESTR]
\end{alltt}
\item Weak equivalence is substitutive under relabelling operator:
\begin{alltt}
\HOLTokenTurnstile{} \HOLFreeVar{E} \HOLSymConst{\ensuremath{\approx}} \HOLFreeVar{E\sp{\prime}} \HOLSymConst{\HOLTokenImp{}} \HOLSymConst{\HOLTokenForall{}}\HOLBoundVar{rf}. \HOLConst{relab} \HOLFreeVar{E} \HOLBoundVar{rf} \HOLSymConst{\ensuremath{\approx}} \HOLConst{relab} \HOLFreeVar{E\sp{\prime}} \HOLBoundVar{rf}\hfill[WEAK_EQUIV_SUBST_RELAB]
\end{alltt}
\end{enumerate}
\end{proposition}

Finally, we have proved that, strong equivalence implies weak
equivalence:
\begin{theorem}{(Strong equivalence implies weak equivalence)}
\begin{alltt}
\HOLTokenTurnstile{} \HOLFreeVar{E} \HOLSymConst{\ensuremath{\sim}} \HOLFreeVar{E\sp{\prime}} \HOLSymConst{\HOLTokenImp{}} \HOLFreeVar{E} \HOLSymConst{\ensuremath{\approx}} \HOLFreeVar{E\sp{\prime}}\hfill[STRONG_IMP_WEAK_EQUIV]
\end{alltt}
\end{theorem}

Here we omit all the algebraic laws for weak equivalence, because they
were all easily derived from the corresponding algebraic laws for
strong equivalence, except for the following $\tau$-law:
\begin{theorem}{The $\tau$-law for weak equivalence)}
\begin{alltt}
\HOLTokenTurnstile{} \HOLSymConst{\ensuremath{\tau}}\HOLSymConst{..}\HOLFreeVar{E} \HOLSymConst{\ensuremath{\approx}} \HOLFreeVar{E}\hfill[TAU_WEAK]
\end{alltt}
\end{theorem}

\section{Observation congruence}

The concept of rooted weak bisimulation equivalence (also namsed
\emph{observation congruence}) is an ``obvious fix'' to convert weak
bisimulation equivalence into a congruence. Its definition is not
recursive but based on the definition of weak equivalence:
\begin{definition}{(Observation congruence)}
Two CCS processes are observation congruence if and only if for any
transition from one of them, there's a responding weak transition from
the other, and the resulting two sub-processes are weak equivalence:
\begin{alltt}
\HOLTokenTurnstile{} \HOLFreeVar{E} \HOLSymConst{\HOLTokenObsCongr} \HOLFreeVar{E\sp{\prime}} \HOLSymConst{\HOLTokenEquiv{}}
   \HOLSymConst{\HOLTokenForall{}}\HOLBoundVar{u}.
     (\HOLSymConst{\HOLTokenForall{}}\HOLBoundVar{E\sb{\mathrm{1}}}. \HOLFreeVar{E} --\HOLBoundVar{u}-> \HOLBoundVar{E\sb{\mathrm{1}}} \HOLSymConst{\HOLTokenImp{}} \HOLSymConst{\HOLTokenExists{}}\HOLBoundVar{E\sb{\mathrm{2}}}. \HOLFreeVar{E\sp{\prime}} ==\HOLBoundVar{u}=>> \HOLBoundVar{E\sb{\mathrm{2}}} \HOLSymConst{\HOLTokenConj{}} \HOLBoundVar{E\sb{\mathrm{1}}} \HOLSymConst{\ensuremath{\approx}} \HOLBoundVar{E\sb{\mathrm{2}}}) \HOLSymConst{\HOLTokenConj{}}
     \HOLSymConst{\HOLTokenForall{}}\HOLBoundVar{E\sb{\mathrm{2}}}. \HOLFreeVar{E\sp{\prime}} --\HOLBoundVar{u}-> \HOLBoundVar{E\sb{\mathrm{2}}} \HOLSymConst{\HOLTokenImp{}} \HOLSymConst{\HOLTokenExists{}}\HOLBoundVar{E\sb{\mathrm{1}}}. \HOLFreeVar{E} ==\HOLBoundVar{u}=>> \HOLBoundVar{E\sb{\mathrm{1}}} \HOLSymConst{\HOLTokenConj{}} \HOLBoundVar{E\sb{\mathrm{1}}} \HOLSymConst{\ensuremath{\approx}} \HOLBoundVar{E\sb{\mathrm{2}}}\hfill[OBS_CONGR]
\end{alltt}
\end{definition}

By observing the differences between the definition of observation
equivalence (weak equivalence) and congruence, we can see that,
observation equivalence requires a little more: for each $\tau$-transition
from one process, the other process must response with at least one
$\tau$-transition. Thus what's immediately
proven is the following two theorems:
\begin{theorem}{(Observation congruence implies observation
    equivalence)}
\begin{alltt}
\HOLTokenTurnstile{} \HOLFreeVar{E} \HOLSymConst{\HOLTokenObsCongr} \HOLFreeVar{E\sp{\prime}} \HOLSymConst{\HOLTokenImp{}} \HOLFreeVar{E} \HOLSymConst{\ensuremath{\approx}} \HOLFreeVar{E\sp{\prime}}\hfill[OBS_CONGR_IMP_WEAK_EQUIV]
\end{alltt}
\end{theorem}

\begin{theorem}{(Observation equivalence on stable agents implies
    observation congruence)}
\begin{alltt}
\HOLTokenTurnstile{} \HOLFreeVar{E} \HOLSymConst{\ensuremath{\approx}} \HOLFreeVar{E\sp{\prime}} \HOLSymConst{\HOLTokenConj{}} \HOLConst{STABLE} \HOLFreeVar{E} \HOLSymConst{\HOLTokenConj{}} \HOLConst{STABLE} \HOLFreeVar{E\sp{\prime}} \HOLSymConst{\HOLTokenImp{}} \HOLFreeVar{E} \HOLSymConst{\HOLTokenObsCongr} \HOLFreeVar{E\sp{\prime}}\hfill[WEAK_EQUIV_STABLE_IMP_CONGR]
\end{alltt}
\end{theorem}

Surprisingly, it's not trivial to prove that, the observation
equivalence is indeed an equivalence relation. The reflexitivy and
symmetry are trivial:
\begin{proposition}{(The reflexitivy and symmetry of observation congruence)}
\begin{alltt}
\HOLTokenTurnstile{} \HOLFreeVar{E} \HOLSymConst{\HOLTokenObsCongr} \HOLFreeVar{E}\hfill[OBS_CONGR_REFL]
\HOLTokenTurnstile{} \HOLFreeVar{E} \HOLSymConst{\HOLTokenObsCongr} \HOLFreeVar{E\sp{\prime}} \HOLSymConst{\HOLTokenImp{}} \HOLFreeVar{E\sp{\prime}} \HOLSymConst{\HOLTokenObsCongr} \HOLFreeVar{E}\hfill[OBS_CONGR_SYM]
\end{alltt}
\end{proposition}
But the transitivity is hard to prove.\footnote{Actually it's not
  proven in the old work, the formal proofs that we did in this
  project is completely new.}  Our proof here is based on the following
lemmas:
\begin{lemma}
If two processes $E$ and $E'$ are observation congruence, then for any
EPS transition coming from $E$, there's a corresponding EPS transition
from $E'$, and the resulting two subprocesses are weakly equivalent:
\begin{alltt}
\HOLTokenTurnstile{} \HOLFreeVar{E} \HOLSymConst{\HOLTokenObsCongr} \HOLFreeVar{E\sp{\prime}} \HOLSymConst{\HOLTokenImp{}} \HOLSymConst{\HOLTokenForall{}}\HOLBoundVar{E\sb{\mathrm{1}}}. \HOLFreeVar{E} \HOLSymConst{\HOLTokenEPS} \HOLBoundVar{E\sb{\mathrm{1}}} \HOLSymConst{\HOLTokenImp{}} \HOLSymConst{\HOLTokenExists{}}\HOLBoundVar{E\sb{\mathrm{2}}}. \HOLFreeVar{E\sp{\prime}} \HOLSymConst{\HOLTokenEPS} \HOLBoundVar{E\sb{\mathrm{2}}} \HOLSymConst{\HOLTokenConj{}} \HOLBoundVar{E\sb{\mathrm{1}}} \HOLSymConst{\ensuremath{\approx}} \HOLBoundVar{E\sb{\mathrm{2}}}\hfill[OBS_CONGR_EPS]
\end{alltt}
\end{lemma}
\begin{proof}
By (right) induction\footnote{The induction theorem used here is \texttt{EPS_ind_right}.} on the number of $\tau$ in the EPS transition of $E$. In
the base case, there's no $\tau$ at all, the $E$ transites to
itself. And in this case $E$' can respond with itself, which is also
an EPS transition:
\begin{displaymath}
\xymatrix{
  {E} \ar@{.}[r]^{\approx^c} \ar@{-}[d]^{=} & {E'} \ar@{-}[d]^{=} \\
  {E} \ar@{.}[r]^{\approx} & {E'}
}
\end{displaymath}

For the induction case, suppose the proposition is true for zero or more $\tau$
transitions except for the last step, that's, $\forall E, \exists E1,
E2$, such that \HOLinline{\HOLFreeVar{E} \HOLSymConst{\HOLTokenEPS} \HOLFreeVar{E\sb{\mathrm{1}}}}, \HOLinline{\HOLFreeVar{E\sp{\prime}} \HOLSymConst{\HOLTokenEPS} \HOLFreeVar{E\sb{\mathrm{2}}}} and
\HOLinline{\HOLFreeVar{E\sb{\mathrm{1}}} \HOLSymConst{\ensuremath{\approx}} \HOLFreeVar{E\sb{\mathrm{2}}}}. Now by definition of weak equivalence, if
\HOLinline{\HOLFreeVar{E\sb{\mathrm{1}}} --\HOLSymConst{\ensuremath{\tau}}-> \HOLFreeVar{E\sb{\mathrm{1}}\sp{\prime}}} then there exists $E2'$ such that \HOLinline{\HOLFreeVar{E\sb{\mathrm{2}}} \HOLSymConst{\HOLTokenEPS} \HOLFreeVar{E\sb{\mathrm{2}}\sp{\prime}}} and \HOLinline{\HOLFreeVar{E\sb{\mathrm{1}}\sp{\prime}} \HOLSymConst{\ensuremath{\approx}} \HOLFreeVar{E\sb{\mathrm{2}}\sp{\prime}}}. Then by transitivity of EPS,
we have \HOLinline{\HOLFreeVar{E\sp{\prime}} \HOLSymConst{\HOLTokenEPS} \HOLFreeVar{E\sb{\mathrm{2}}} \HOLSymConst{\HOLTokenConj{}} \HOLFreeVar{E\sb{\mathrm{2}}} \HOLSymConst{\HOLTokenEPS} \HOLFreeVar{E\sb{\mathrm{2}}\sp{\prime}} \HOLSymConst{\HOLTokenImp{}} \HOLFreeVar{E\sp{\prime}} \HOLSymConst{\HOLTokenEPS} \HOLFreeVar{E\sb{\mathrm{2}}\sp{\prime}}}, thus $E_2'$ is a valid response required by
observation congruence:
\begin{displaymath}
\xymatrix{
{E} \ar@{.}[r]^{\approx^c} \ar@{=>}[d]^{\epsilon} & {E'}
\ar@{=>}[d]^{\epsilon} \ar@/^4ex/[dd]^{\epsilon} \\
{\forall E_1} \ar@{.}[r]^{\approx} \ar[d]^{\tau} & {\forall E_2} \ar@{=>}[d]^{\epsilon}
\\
{\forall E_1'} \ar@{.}[r]^{\approx} & {\exists E_2'}
}
\end{displaymath}
\qed
\end{proof}

\begin{lemma}
If two processes $E$ and $E'$ are observation congruence, then for any
weak transition coming from $E$, there's a corresponding weak
transition from $E'$, and the resulting two subprocesses are weakly
equivalent:
\begin{alltt}
\HOLTokenTurnstile{} \HOLFreeVar{E} \HOLSymConst{\HOLTokenObsCongr} \HOLFreeVar{E\sp{\prime}} \HOLSymConst{\HOLTokenImp{}} \HOLSymConst{\HOLTokenForall{}}\HOLBoundVar{u} \HOLBoundVar{E\sb{\mathrm{1}}}. \HOLFreeVar{E} ==\HOLBoundVar{u}=>> \HOLBoundVar{E\sb{\mathrm{1}}} \HOLSymConst{\HOLTokenImp{}} \HOLSymConst{\HOLTokenExists{}}\HOLBoundVar{E\sb{\mathrm{2}}}. \HOLFreeVar{E\sp{\prime}} ==\HOLBoundVar{u}=>> \HOLBoundVar{E\sb{\mathrm{2}}} \HOLSymConst{\HOLTokenConj{}} \HOLBoundVar{E\sb{\mathrm{1}}} \HOLSymConst{\ensuremath{\approx}} \HOLBoundVar{E\sb{\mathrm{2}}}
\end{alltt}
\end{lemma}
\begin{proof}{(sketch}
Consider the two cases when the action is $\tau$ or not
$\tau$. For all weak $\tau$-transitions coming from $E$, the
observation congruence requires that there's at least one $\tau$
following $E'$ and the resulting two sub-processes, say $E_1'$ and
$E_2$ are weak equivalence. Then the desired responses can be found by
using a similar existence lemma for weak equivalence:
\begin{displaymath}
\xymatrix{
{E} \ar@{.}[r]^{\approx^c} \ar[d]^{\tau} \ar@/^-4ex/[dd]^{\tau} & {E'}
\ar@{=>}[d]^{\tau} \ar@/^4ex/[dd]^{\tau} \\
{\exists E_1'} \ar@{.}[r]^{\approx} \ar@{=>}[d]^{\epsilon} & {\exists E_2}
\ar@{=>}[d]^{\epsilon} \\
{\forall E_1} \ar@{.}[r]^{\approx} & {\exists E_2'}
}
\end{displaymath}

For all the non-$\tau$ weak transitions from $E$, the proof follows
from previous lemma and a similar existence lemma for weak
equivalence. The following figure is a sketch for the proof of this
case:
\begin{displaymath}
\xymatrix{
{E} \ar@{.}[r]^{\approx^c} \ar@{=>}[d]^{\epsilon}
\ar@/^-4ex/[ddd]^{\forall L}
& {E'} \ar@{=>}[d]^{\epsilon} \ar@/^4ex/[ddd]^{L} \\
{\exists E_1'} \ar@{.}[r]^{\approx} \ar[d]^{L} & {\exists E_2'} \ar@{=>}[d]^{L} \\
{\exists E_2} \ar@{=>}[d]^{\epsilon} \ar@{.}[r]^{\approx} & {\exists E_2''}
\ar@{=>}[d]^{\epsilon} \\
{\forall E_1} \ar@{.}[r]^{\approx} & {\exists E2'''}
}
\end{displaymath}
In the previous figure, the existence of $E_2'$ follows by previous lemma, the existence of
$E_2''$ follows by the definition of weak equivalence, and the
existence of $E_2'''$ follows by the next existence lemma of weak equivalence.
\qed
\end{proof}

The existence lemma for weak equivalences that we mentioned in
previous proof is the following one:
\begin{lemma}
\begin{alltt}
\HOLTokenTurnstile{} \HOLFreeVar{E} \HOLSymConst{\HOLTokenEPS} \HOLFreeVar{E\sb{\mathrm{1}}} \HOLSymConst{\HOLTokenImp{}}
   \HOLSymConst{\HOLTokenForall{}}\HOLBoundVar{Wbsm} \HOLBoundVar{E\sp{\prime}}.
     \HOLConst{WEAK_BISIM} \HOLBoundVar{Wbsm} \HOLSymConst{\HOLTokenConj{}} \HOLBoundVar{Wbsm} \HOLFreeVar{E} \HOLBoundVar{E\sp{\prime}} \HOLSymConst{\HOLTokenImp{}} \HOLSymConst{\HOLTokenExists{}}\HOLBoundVar{E\sb{\mathrm{2}}}. \HOLBoundVar{E\sp{\prime}} \HOLSymConst{\HOLTokenEPS} \HOLBoundVar{E\sb{\mathrm{2}}} \HOLSymConst{\HOLTokenConj{}} \HOLBoundVar{Wbsm} \HOLFreeVar{E\sb{\mathrm{1}}} \HOLBoundVar{E\sb{\mathrm{2}}}
\end{alltt}
\end{lemma}

Now we prove the transitivity of observation congruence:
\begin{theorem}{(Transitivity of Observation Congruence)}
\begin{alltt}
\HOLTokenTurnstile{} \HOLFreeVar{E} \HOLSymConst{\HOLTokenObsCongr} \HOLFreeVar{E\sp{\prime}} \HOLSymConst{\HOLTokenConj{}} \HOLFreeVar{E\sp{\prime}} \HOLSymConst{\HOLTokenObsCongr} \HOLFreeVar{E\sp{\prime\prime}} \HOLSymConst{\HOLTokenImp{}} \HOLFreeVar{E} \HOLSymConst{\HOLTokenObsCongr} \HOLFreeVar{E\sp{\prime\prime}}\hfill[OBS_CONGR_TRANS]
\end{alltt}
\end{theorem}
\begin{proof}
Suppose \HOLinline{\HOLFreeVar{E} \HOLSymConst{\HOLTokenObsCongr} \HOLFreeVar{E\sp{\prime}}} and \HOLinline{\HOLFreeVar{E\sp{\prime}} \HOLSymConst{\HOLTokenObsCongr} \HOLFreeVar{E\sp{\prime\prime}}}, we're
going to prove \HOLinline{\HOLFreeVar{E} \HOLSymConst{\HOLTokenObsCongr} \HOLFreeVar{E\sp{\prime\prime}}} by checking directly the
definition of observation congruence.

For any $u$ and $E_1$ which satisfy \HOLinline{\HOLFreeVar{E} --\HOLFreeVar{u}-> \HOLFreeVar{E\sb{\mathrm{1}}}}, by definition of observation congruence, there exists $E_2$
such that \HOLinline{\HOLFreeVar{E\sp{\prime}} ==\HOLFreeVar{u}=>> \HOLFreeVar{E\sb{\mathrm{2}}}} with \HOLinline{\HOLFreeVar{E\sb{\mathrm{1}}} \HOLSymConst{\ensuremath{\approx}} \HOLFreeVar{E\sb{\mathrm{2}}}}. By above Lemma 2, there exists
another $E_3$ such that \HOLinline{\HOLFreeVar{E\sp{\prime\prime}} ==\HOLFreeVar{u}=>> \HOLFreeVar{E\sb{\mathrm{3}}}} with
\HOLinline{\HOLFreeVar{E\sb{\mathrm{2}}} \HOLSymConst{\ensuremath{\approx}} \HOLFreeVar{E\sb{\mathrm{3}}}}. By the already proven transitivity of weak
equivalence, \HOLinline{\HOLFreeVar{E\sb{\mathrm{1}}} \HOLSymConst{\ensuremath{\approx}} \HOLFreeVar{E\sb{\mathrm{3}}}}, thus $E_3$ is the required
process which satisfies the definition of observation
congruence. This proves the first part. The other part is completely symmetric.
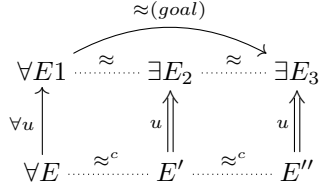
\begin{displaymath}
\xymatrix{
{\forall E1} \ar@{.}[r]^{\approx} \ar@/^4ex/[rr]^{\approx (goal)} &
{\exists E_2}
\ar@{.}[r]^{\approx} & {\exists E_3} \\
{\forall E} \ar[u]^{\forall u} \ar@{.}[r]^{\approx^c} & {E'} \ar@{=>}[u]^{u}
\ar@{.}[r]^{\approx^c} & {E''} \ar@{=>}[u]^{u}
}
\end{displaymath}
\qed
\end{proof}

Then we have proved the substitutivity of observation congruence under
various CCS process operators:
\begin{proposition}
\begin{enumerate}
\item Observation congruence is substitutive under the prefix
  operator:
\begin{alltt}
\HOLTokenTurnstile{} \HOLFreeVar{E} \HOLSymConst{\HOLTokenObsCongr} \HOLFreeVar{E\sp{\prime}} \HOLSymConst{\HOLTokenImp{}} \HOLSymConst{\HOLTokenForall{}}\HOLBoundVar{u}. \HOLBoundVar{u}\HOLSymConst{..}\HOLFreeVar{E} \HOLSymConst{\HOLTokenObsCongr} \HOLBoundVar{u}\HOLSymConst{..}\HOLFreeVar{E\sp{\prime}}\hfill[OBS_CONGR_SUBST_PREFIX]
\end{alltt}
\item Observation congruence is substitutive under binary summation:
\begin{alltt}
\HOLTokenTurnstile{} \HOLFreeVar{p} \HOLSymConst{\HOLTokenObsCongr} \HOLFreeVar{q} \HOLSymConst{\HOLTokenConj{}} \HOLFreeVar{r} \HOLSymConst{\HOLTokenObsCongr} \HOLFreeVar{s} \HOLSymConst{\HOLTokenImp{}} \HOLFreeVar{p} \HOLSymConst{+} \HOLFreeVar{r} \HOLSymConst{\HOLTokenObsCongr} \HOLFreeVar{q} \HOLSymConst{+} \HOLFreeVar{s}\hfill[OBS_CONGR_PRESD_BY_SUM]
\end{alltt}
\item Observation congruence is preserved by parallel composition:
\begin{alltt}
\HOLTokenTurnstile{} \HOLFreeVar{E\sb{\mathrm{1}}} \HOLSymConst{\HOLTokenObsCongr} \HOLFreeVar{E\sb{\mathrm{1}}\sp{\prime}} \HOLSymConst{\HOLTokenConj{}} \HOLFreeVar{E\sb{\mathrm{2}}} \HOLSymConst{\HOLTokenObsCongr} \HOLFreeVar{E\sb{\mathrm{2}}\sp{\prime}} \HOLSymConst{\HOLTokenImp{}} \HOLFreeVar{E\sb{\mathrm{1}}} \HOLSymConst{||} \HOLFreeVar{E\sb{\mathrm{2}}} \HOLSymConst{\HOLTokenObsCongr} \HOLFreeVar{E\sb{\mathrm{1}}\sp{\prime}} \HOLSymConst{||} \HOLFreeVar{E\sb{\mathrm{2}}\sp{\prime}}\hfill[OBS_CONGR_PRESD_BY_PAR]
\end{alltt}
\item Observation congruence is substitutive under the restriction
  operator:
\begin{alltt}
\HOLTokenTurnstile{} \HOLFreeVar{E} \HOLSymConst{\HOLTokenObsCongr} \HOLFreeVar{E\sp{\prime}} \HOLSymConst{\HOLTokenImp{}} \HOLSymConst{\HOLTokenForall{}}\HOLBoundVar{L}. \HOLSymConst{\ensuremath{\nu}} \HOLBoundVar{L} \HOLFreeVar{E} \HOLSymConst{\HOLTokenObsCongr} \HOLSymConst{\ensuremath{\nu}} \HOLBoundVar{L} \HOLFreeVar{E\sp{\prime}}\hfill[OBS_CONGR_SUBST_RESTR]
\end{alltt}
\item Observation congruence is substitutive under the relabeling
  operator:
\begin{alltt}
\HOLTokenTurnstile{} \HOLFreeVar{E} \HOLSymConst{\HOLTokenObsCongr} \HOLFreeVar{E\sp{\prime}} \HOLSymConst{\HOLTokenImp{}} \HOLSymConst{\HOLTokenForall{}}\HOLBoundVar{rf}. \HOLConst{relab} \HOLFreeVar{E} \HOLBoundVar{rf} \HOLSymConst{\HOLTokenObsCongr} \HOLConst{relab} \HOLFreeVar{E\sp{\prime}} \HOLBoundVar{rf}\hfill[OBS_CONGR_SUBST_RELAB]
\end{alltt}
\end{enumerate}
\end{proposition}

Finally, like the case for weak equivalence, we can easily prove the
relationship between strong equivalence and observation congruence:
\begin{theorem}{(Strong equivalence implies observation congruence)}
\begin{alltt}
\HOLTokenTurnstile{} \HOLFreeVar{E} \HOLSymConst{\ensuremath{\sim}} \HOLFreeVar{E\sp{\prime}} \HOLSymConst{\HOLTokenImp{}} \HOLFreeVar{E} \HOLSymConst{\HOLTokenObsCongr} \HOLFreeVar{E\sp{\prime}}\hfill[STRONG_IMP_OBS_CONGR]
\end{alltt}
\end{theorem}

With this result, all algebraic laws for observation congruence can be
derived from the corresponding algebraic laws of strong
equivalence. Here we omit these theorems, except for the following four
$\tau$-laws:
\begin{theorem}{(The $\tau$-laws for observation congruence)}
\begin{alltt}
\HOLTokenTurnstile{} \HOLFreeVar{u}\HOLSymConst{..}\HOLSymConst{\ensuremath{\tau}}\HOLSymConst{..}\HOLFreeVar{E} \HOLSymConst{\HOLTokenObsCongr} \HOLFreeVar{u}\HOLSymConst{..}\HOLFreeVar{E}\hfill[TAU1]
\HOLTokenTurnstile{} \HOLFreeVar{E} \HOLSymConst{+} \HOLSymConst{\ensuremath{\tau}}\HOLSymConst{..}\HOLFreeVar{E} \HOLSymConst{\HOLTokenObsCongr} \HOLSymConst{\ensuremath{\tau}}\HOLSymConst{..}\HOLFreeVar{E}\hfill[TAU2]
\HOLTokenTurnstile{} \HOLFreeVar{u}\HOLSymConst{..}(\HOLFreeVar{E} \HOLSymConst{+} \HOLSymConst{\ensuremath{\tau}}\HOLSymConst{..}\HOLFreeVar{E\sp{\prime}}) \HOLSymConst{+} \HOLFreeVar{u}\HOLSymConst{..}\HOLFreeVar{E\sp{\prime}} \HOLSymConst{\HOLTokenObsCongr} \HOLFreeVar{u}\HOLSymConst{..}(\HOLFreeVar{E} \HOLSymConst{+} \HOLSymConst{\ensuremath{\tau}}\HOLSymConst{..}\HOLFreeVar{E\sp{\prime}})\hfill[TAU3]
\HOLTokenTurnstile{} \HOLFreeVar{E} \HOLSymConst{+} \HOLSymConst{\ensuremath{\tau}}\HOLSymConst{..}(\HOLFreeVar{E\sp{\prime}} \HOLSymConst{+} \HOLFreeVar{E}) \HOLSymConst{\HOLTokenObsCongr} \HOLSymConst{\ensuremath{\tau}}\HOLSymConst{..}(\HOLFreeVar{E\sp{\prime}} \HOLSymConst{+} \HOLFreeVar{E})\hfill[TAU_STRAT]
\end{alltt}
\end{theorem}

\section{Deng lemma and Hennessy lemma}

The relationship between weak equivalence and observation congruence
was an interesting research topic, and there're many deep lemmas
related. In this project, we have proved two such deep lemmas. The
first one is the following Deng Lemma (for weak
bisimularity\footnote{The original Deng lemma is for another kind of
  equivalence relation called \emph{rooted branching bisimularity},
  which is not touched in this project.}):
\begin{theorem}{(Deng lemma for weak bisimilarity)}
If \HOLinline{\HOLFreeVar{p} \HOLSymConst{\ensuremath{\approx}} \HOLFreeVar{q}}, then one of the following three cases
holds:
\begin{enumerate}
\item $\exists p'$ such that \HOLinline{\HOLFreeVar{p} --\HOLSymConst{\ensuremath{\tau}}-> \HOLFreeVar{p\sp{\prime}}} and
  \HOLinline{\HOLFreeVar{p\sp{\prime}} \HOLSymConst{\ensuremath{\approx}} \HOLFreeVar{q}}, or
\item $\exists q'$ such that \HOLinline{\HOLFreeVar{q} --\HOLSymConst{\ensuremath{\tau}}-> \HOLFreeVar{q\sp{\prime}}} and
  \HOLinline{\HOLFreeVar{p} \HOLSymConst{\ensuremath{\approx}} \HOLFreeVar{q\sp{\prime}}}, or
\item \HOLinline{\HOLFreeVar{p} \HOLSymConst{\HOLTokenObsCongr} \HOLFreeVar{q}}.
\end{enumerate}
\begin{alltt}
\HOLTokenTurnstile{} \HOLFreeVar{p} \HOLSymConst{\ensuremath{\approx}} \HOLFreeVar{q} \HOLSymConst{\HOLTokenImp{}}
   (\HOLSymConst{\HOLTokenExists{}}\HOLBoundVar{p\sp{\prime}}. \HOLFreeVar{p} --\HOLSymConst{\ensuremath{\tau}}-> \HOLBoundVar{p\sp{\prime}} \HOLSymConst{\HOLTokenConj{}} \HOLBoundVar{p\sp{\prime}} \HOLSymConst{\ensuremath{\approx}} \HOLFreeVar{q}) \HOLSymConst{\HOLTokenDisj{}} (\HOLSymConst{\HOLTokenExists{}}\HOLBoundVar{q\sp{\prime}}. \HOLFreeVar{q} --\HOLSymConst{\ensuremath{\tau}}-> \HOLBoundVar{q\sp{\prime}} \HOLSymConst{\HOLTokenConj{}} \HOLFreeVar{p} \HOLSymConst{\ensuremath{\approx}} \HOLBoundVar{q\sp{\prime}}) \HOLSymConst{\HOLTokenDisj{}}
   \HOLFreeVar{p} \HOLSymConst{\HOLTokenObsCongr} \HOLFreeVar{q}\hfill[DENG_LEMMA]
\end{alltt}
\end{theorem}
\begin{proof}
Actually there's no need to consider thee difference cases. Using the
logical tautology \HOLinline{(\HOLSymConst{\HOLTokenNeg{}}\HOLFreeVar{P} \HOLSymConst{\HOLTokenConj{}} \HOLSymConst{\HOLTokenNeg{}}\HOLFreeVar{Q} \HOLSymConst{\HOLTokenImp{}} \HOLFreeVar{R}) \HOLSymConst{\HOLTokenImp{}} \HOLFreeVar{P} \HOLSymConst{\HOLTokenDisj{}} \HOLFreeVar{Q} \HOLSymConst{\HOLTokenDisj{}} \HOLFreeVar{R}},
the theorem can be reduced to the following goal:
\begin{quote}
Prove \HOLinline{\HOLFreeVar{p} \HOLSymConst{\HOLTokenObsCongr} \HOLFreeVar{q}}, with the following three assumptions:
\begin{enumerate}
\item \HOLinline{\HOLFreeVar{p} \HOLSymConst{\ensuremath{\approx}} \HOLFreeVar{q}}
\item \HOLinline{\HOLSymConst{\HOLTokenNeg{}}\HOLSymConst{\HOLTokenExists{}}\HOLBoundVar{p\sp{\prime}}. \HOLFreeVar{p} --\HOLSymConst{\ensuremath{\tau}}-> \HOLBoundVar{p\sp{\prime}} \HOLSymConst{\HOLTokenConj{}} \HOLBoundVar{p\sp{\prime}} \HOLSymConst{\ensuremath{\approx}} \HOLFreeVar{q}}
\item \HOLinline{\HOLSymConst{\HOLTokenNeg{}}\HOLSymConst{\HOLTokenExists{}}\HOLBoundVar{q\sp{\prime}}. \HOLFreeVar{q} --\HOLSymConst{\ensuremath{\tau}}-> \HOLBoundVar{q\sp{\prime}} \HOLSymConst{\HOLTokenConj{}} \HOLFreeVar{p} \HOLSymConst{\ensuremath{\approx}} \HOLBoundVar{q\sp{\prime}}}
\end{enumerate}
\end{quote}

Now we check the definition of observation congruence: for any
transition from $p$, say \HOLinline{\HOLFreeVar{p} --\HOLFreeVar{u}-> \HOLFreeVar{E\sb{\mathrm{1}}}}, consider the cases when
$u = \tau$ and $u \neq \tau$:
\begin{enumerate}
\item If $u = \tau$, then by \HOLinline{\HOLFreeVar{p} \HOLSymConst{\ensuremath{\approx}} \HOLFreeVar{q}} and the definition
  of weak equivalence, there exists $E_2$ such that \HOLinline{\HOLFreeVar{q} \HOLSymConst{\HOLTokenEPS} \HOLFreeVar{E\sb{\mathrm{2}}}}
  and \HOLinline{\HOLFreeVar{E\sb{\mathrm{1}}} \HOLSymConst{\ensuremath{\approx}} \HOLFreeVar{E\sb{\mathrm{2}}}}. But by assumption we know $q \neq E2$, thus
  \HOLinline{\HOLFreeVar{q} \HOLSymConst{\HOLTokenEPS} \HOLFreeVar{E\sb{\mathrm{2}}}} contains at least one $\tau$-transition, thus is
  actually \HOLinline{\HOLFreeVar{q} ==\HOLSymConst{\ensuremath{\tau}}=>> \HOLFreeVar{E\sb{\mathrm{2}}}}, which is required by the definition
  of observation congruence for \HOLinline{\HOLFreeVar{p} \HOLSymConst{\ensuremath{\approx}} \HOLFreeVar{q}}.
\begin{displaymath}
\xymatrix{
{p} \ar@{.}[r]^{\approx^c} \ar[d]^{\tau} & {q} \ar@{=>}[d]^{\epsilon
  (\tau)} \\
{\forall E_1} \ar@{.}[ur]^{\not\approx} \ar@{.}[r]^{\approx} &
{\exists E_2}
}
\end{displaymath}
\item If $u = L$, then the requirement of observation congruence is
  directly satisfied.
\end{enumerate}
The other direction is completely symmetric.\qed
\end{proof}

Now we start to prove Hennessy Lemma:
\begin{theorem}{(Hennessy Lemma)}
For any processes $p$ and $q$, \HOLinline{\HOLFreeVar{p} \HOLSymConst{\ensuremath{\approx}} \HOLFreeVar{q}} if and only if
(\HOLinline{\HOLFreeVar{p} \HOLSymConst{\HOLTokenObsCongr} \HOLFreeVar{q}} or \HOLinline{\HOLFreeVar{p} \HOLSymConst{\HOLTokenObsCongr} \HOLSymConst{\ensuremath{\tau}}\HOLSymConst{..}\HOLFreeVar{q}} or
\HOLinline{\HOLSymConst{\ensuremath{\tau}}\HOLSymConst{..}\HOLFreeVar{p} \HOLSymConst{\HOLTokenObsCongr} \HOLFreeVar{q}}):
\begin{alltt}
\HOLTokenTurnstile{} \HOLFreeVar{p} \HOLSymConst{\ensuremath{\approx}} \HOLFreeVar{q} \HOLSymConst{\HOLTokenEquiv{}} \HOLFreeVar{p} \HOLSymConst{\HOLTokenObsCongr} \HOLFreeVar{q} \HOLSymConst{\HOLTokenDisj{}} \HOLFreeVar{p} \HOLSymConst{\HOLTokenObsCongr} \HOLSymConst{\ensuremath{\tau}}\HOLSymConst{..}\HOLFreeVar{q} \HOLSymConst{\HOLTokenDisj{}} \HOLSymConst{\ensuremath{\tau}}\HOLSymConst{..}\HOLFreeVar{p} \HOLSymConst{\HOLTokenObsCongr} \HOLFreeVar{q}\hfill[HENNESSY_LEMMA]
\end{alltt}
\end{theorem}

\begin{proof}
The ``if'' part (from right to left) can be easily derived by applying
\texttt{OBS_CONGR_IMP_WEAK_EQUIV}, \texttt{TAU_WEAK},
\texttt{WEAK_EQUIV_SYM} and \texttt{WEAK_EQUIV_TRANS}. We'll focus on
the hard ``only if'' part (from left to right). The proof represent
here is slightly simplier than the one in \cite{Gorrieri:2015jt}, but
the idea is the same. The proof is based on creative case analysis.

If there exists an $E$ such that \HOLinline{\HOLFreeVar{p} --\HOLSymConst{\ensuremath{\tau}}-> \HOLFreeVar{E} \HOLSymConst{\HOLTokenConj{}} \HOLFreeVar{E} \HOLSymConst{\ensuremath{\approx}} \HOLFreeVar{q}} than we can prove that \HOLinline{\HOLFreeVar{p} \HOLSymConst{\HOLTokenObsCongr} \HOLSymConst{\ensuremath{\tau}}\HOLSymConst{..}\HOLFreeVar{q}} by
expanding \HOLinline{\HOLFreeVar{p} \HOLSymConst{\ensuremath{\approx}} \HOLFreeVar{q}} by \texttt{OBS_PROPERTY_STAR}. The
other needed theorems are the definition of weak transition,
\texttt{EPS_REFL}, SOS rule \texttt{PREFIX} and \texttt{TRANS_PREFIX},
\texttt{TAU_PREFIX_WEAK_TRANS} and \texttt{TRANS_IMP_WEAK_TRANS}.

If there's no $E$ such that \HOLinline{\HOLFreeVar{p} --\HOLSymConst{\ensuremath{\tau}}-> \HOLFreeVar{E} \HOLSymConst{\HOLTokenConj{}} \HOLFreeVar{E} \HOLSymConst{\ensuremath{\approx}} \HOLFreeVar{q}}, we can further check if there exist an $E$ such that
\HOLinline{\HOLFreeVar{q} --\HOLSymConst{\ensuremath{\tau}}-> \HOLFreeVar{E} \HOLSymConst{\HOLTokenConj{}} \HOLFreeVar{p} \HOLSymConst{\ensuremath{\approx}} \HOLFreeVar{E}}, and in this case we can prove
\HOLinline{\HOLSymConst{\ensuremath{\tau}}\HOLSymConst{..}\HOLFreeVar{p} \HOLSymConst{\HOLTokenObsCongr} \HOLFreeVar{q}} in the same way as the above case.

Otherwise we got exactly the same condition as in Deng Lemma
(after the initial goal reduced in the previous proof), and in this
case we can directly prove that \HOLinline{\HOLFreeVar{p} \HOLSymConst{\HOLTokenObsCongr} \HOLFreeVar{q}}.
\end{proof}

The purpose of this formal proof has basically shown that, for most
informal proofs in Concurrency Theory which doesn't depend on external
mathematics theories, the author has got the ability to
express it in HOL theorem prover.

\section{The theory of congruence}

The highlight of this project is the formal proofs for various
versions of the ``coarsest congruence contained in weak equivalence'',
\begin{proposition}{(Coarsest congruence contained in $\approx$)}
For any processes $p$ and $q$, \HOLinline{\HOLFreeVar{p} \HOLSymConst{\HOLTokenObsCongr} \HOLFreeVar{q}} if and only if
\HOLinline{\HOLSymConst{\HOLTokenForall{}}\HOLBoundVar{r}. \HOLFreeVar{p} \HOLSymConst{+} \HOLBoundVar{r} \HOLSymConst{\ensuremath{\approx}} \HOLFreeVar{q} \HOLSymConst{+} \HOLBoundVar{r}}.
\end{proposition}

But at first glance, the name of above theorem doesn't make much
sense. To see the nature of above theorem more clearly, here we
represent a rather complete theory about the congruence of CCS. It's
based on contents from \cite{vanGlabbeek:2005ur}.

To formalize the concept of congruence, we need to define ``semantic context''
first. There're multiple solutions, here we have chosen a simple
solution based on $\lambda$-calculus:
\begin{definition}{(Semantic context  of CCS)}
The semantic context (or one-hole context) of CCS is a function
$C[\cdot]$ of type ``\HOLinline{(\ensuremath{\alpha}, \ensuremath{\beta}) \HOLTyOp{CCS} -> (\ensuremath{\alpha}, \ensuremath{\beta}) \HOLTyOp{CCS}}'' recursively defined by following rules:
\begin{quote}
\begin{alltt}
\HOLConst{CONTEXT} (\HOLTokenLambda{}\HOLBoundVar{x}. \HOLBoundVar{x})
\HOLConst{CONTEXT} \HOLFreeVar{c} \HOLSymConst{\HOLTokenImp{}} \HOLConst{CONTEXT} (\HOLTokenLambda{}\HOLBoundVar{t}. \HOLFreeVar{a}\HOLSymConst{..}\HOLFreeVar{c} \HOLBoundVar{t})
\HOLConst{CONTEXT} \HOLFreeVar{c} \HOLSymConst{\HOLTokenImp{}} \HOLConst{CONTEXT} (\HOLTokenLambda{}\HOLBoundVar{t}. \HOLFreeVar{c} \HOLBoundVar{t} \HOLSymConst{+} \HOLFreeVar{x})
\HOLConst{CONTEXT} \HOLFreeVar{c} \HOLSymConst{\HOLTokenImp{}} \HOLConst{CONTEXT} (\HOLTokenLambda{}\HOLBoundVar{t}. \HOLFreeVar{x} \HOLSymConst{+} \HOLFreeVar{c} \HOLBoundVar{t})
\HOLConst{CONTEXT} \HOLFreeVar{c} \HOLSymConst{\HOLTokenImp{}} \HOLConst{CONTEXT} (\HOLTokenLambda{}\HOLBoundVar{t}. \HOLFreeVar{c} \HOLBoundVar{t} \HOLSymConst{||} \HOLFreeVar{x})
\HOLConst{CONTEXT} \HOLFreeVar{c} \HOLSymConst{\HOLTokenImp{}} \HOLConst{CONTEXT} (\HOLTokenLambda{}\HOLBoundVar{t}. \HOLFreeVar{x} \HOLSymConst{||} \HOLFreeVar{c} \HOLBoundVar{t})
\HOLConst{CONTEXT} \HOLFreeVar{c} \HOLSymConst{\HOLTokenImp{}} \HOLConst{CONTEXT} (\HOLTokenLambda{}\HOLBoundVar{t}. \HOLSymConst{\ensuremath{\nu}} \HOLFreeVar{L} (\HOLFreeVar{c} \HOLBoundVar{t}))
\HOLConst{CONTEXT} \HOLFreeVar{c} \HOLSymConst{\HOLTokenImp{}} \HOLConst{CONTEXT} (\HOLTokenLambda{}\HOLBoundVar{t}. \HOLConst{relab} (\HOLFreeVar{c} \HOLBoundVar{t}) \HOLFreeVar{rf})
\end{alltt}
\end{quote}
By repeatedly applying above rules, one can imagine that, a ``hold''
in any CCS term at any depth, can become a $\lambda$-function, and by
calling the function with another CCS term, the hold is filled by that term.
\end{definition}

The notable property of one-hole context is that, the functional
combination of two contexts is still a context:
\begin{proposition}{(The combination of one-hole contexts)}
If both $c_1$ and $c_2$ are one-hole contexts, then $c_1 \circ
c_2$\footnote{$(c_1 \circ c_2) t := c_1 (c_2 t).$} is
still a one-hole context:
\begin{alltt}
\HOLTokenTurnstile{} \HOLConst{CONTEXT} \HOLFreeVar{c\sb{\mathrm{1}}} \HOLSymConst{\HOLTokenConj{}} \HOLConst{CONTEXT} \HOLFreeVar{c\sb{\mathrm{2}}} \HOLSymConst{\HOLTokenImp{}} \HOLConst{CONTEXT} (\HOLFreeVar{c\sb{\mathrm{1}}} \HOLSymConst{\HOLTokenCompose} \HOLFreeVar{c\sb{\mathrm{2}}})
\end{alltt}
\end{proposition}
\begin{proof}
By induction on the first context $c_1$.\qed
\end{proof}

Now we're ready to define the concept of congruence (for CCS):
\begin{definition}{(Congruence of CCS)}
An equivalence relation $\approx$\footnote{The symbol $\approx$ here
  shouldn't be understood as weak equivalence.} on a specific space of
CCS processes is a \emph{congruence} iff for every $n$-ary operator
$f$, one has $g_1 \approx h_1 \wedge \cdots g_n \approx h_n \Rightarrow
f(g_1, \ldots, g_n) \approx f(h_1, \ldots, h_n$. This is the case iff
for every semantic context $C[\cdot]$ on has $g\approx h \Rightarrow
C[g] \approx C[h]$:
\begin{alltt}
\HOLTokenTurnstile{} \HOLConst{congruence} \HOLFreeVar{R} \HOLSymConst{\HOLTokenEquiv{}}
   \HOLSymConst{\HOLTokenForall{}}\HOLBoundVar{x} \HOLBoundVar{y} \HOLBoundVar{ctx}. \HOLConst{CONTEXT} \HOLBoundVar{ctx} \HOLSymConst{\HOLTokenImp{}} \HOLFreeVar{R} \HOLBoundVar{x} \HOLBoundVar{y} \HOLSymConst{\HOLTokenImp{}} \HOLFreeVar{R} (\HOLBoundVar{ctx} \HOLBoundVar{x}) (\HOLBoundVar{ctx} \HOLBoundVar{y})
\end{alltt}
\end{definition}

We can easily prove that, strong equivalence and observation
congruence is indeed a congruence following above definition, using
the substitutability and preserving properties of these relations:
\begin{theorem}
\begin{alltt}
\HOLTokenTurnstile{} \HOLConst{congruence} \HOLConst{STRONG_EQUIV}
\HOLTokenTurnstile{} \HOLConst{congruence} \HOLConst{OBS_CONGR}
\end{alltt}
\end{theorem}

For relations which is not congruence, it's possible to ``convert'' them
into congruence:
\begin{definition}{(Constructing congruences from equivalence relation)}
Given an equivalence relation $\sim$\footnote{The Symbol $\sim$ here
  shouldn't be understood as strong equivalence.}, define $\sim^c$ by:
\begin{alltt}
\HOLTokenTurnstile{} \HOLFreeVar{R}\HOLSymConst{\HOLTokenSupC{}} \HOLSymConst{=} (\HOLTokenLambda{}\HOLBoundVar{g} \HOLBoundVar{h}. \HOLSymConst{\HOLTokenForall{}}\HOLBoundVar{c}. \HOLConst{CONTEXT} \HOLBoundVar{c} \HOLSymConst{\HOLTokenImp{}} \HOLFreeVar{R} (\HOLBoundVar{c} \HOLBoundVar{g}) (\HOLBoundVar{c} \HOLBoundVar{h}))
\end{alltt}
\end{definition}

This new operator on relations has the following three properties:

\begin{proposition}
For all $R$, $R^c$ is a congruence:
\begin{alltt}
\HOLTokenTurnstile{} \HOLConst{congruence} \HOLFreeVar{R}\HOLSymConst{\HOLTokenSupC{}}
\end{alltt}
\end{proposition}
\begin{proof}
By construction, $\sim^c$ is a congruence. For if $g \sim^c h$ and
$D[\cdot]$ is a semantic context, then for every semantic context
$C[\cdot]$ also $C[D[\cdot]]$ is a semantic context, so $\forall
C[\cdot].\;(C[D[g]] \sim C[D[h]])$ and hence $D[g] \sim^c D[h]$.\qed
\end{proof}

\begin{proposition}
For all $R$, $R^c$ is finer than $R$:
\begin{alltt}
\HOLTokenTurnstile{} \HOLFreeVar{R}\HOLSymConst{\HOLTokenSupC{}} \HOLSymConst{\HOLTokenRSubset{}} \HOLFreeVar{R}
\end{alltt}
\end{proposition}
\begin{proof}
The trivial context guarantees that $g \sim^c h \Rightarrow g\sim h$,
so $\sim^c$ is finer than $\sim$.\qed
\end{proof}

\begin{proposition}
For all $R$, $R^c$ is the coarsest congruence finer than $R$, that is,
for any other congruence finer than $R$, it's finer than $R^c$:
\begin{alltt}
\HOLTokenTurnstile{} \HOLConst{congruence} \HOLFreeVar{R\sp{\prime}} \HOLSymConst{\HOLTokenConj{}} \HOLFreeVar{R\sp{\prime}} \HOLSymConst{\HOLTokenRSubset{}} \HOLFreeVar{R} \HOLSymConst{\HOLTokenImp{}} \HOLFreeVar{R\sp{\prime}} \HOLSymConst{\HOLTokenRSubset{}} \HOLFreeVar{R}\HOLSymConst{\HOLTokenSupC{}}
\end{alltt}
\end{proposition}
\begin{proof}
If $\approx$ is any congruence finer than $\sim$, then
\begin{equation}
g \approx h \Rightarrow \forall C[\cdot].\;(C[g] \approx C[h])
\Rightarrow \forall C[\cdot].\;(C[g] \sim C[h]) \Rightarrow g \sim^c h.
\end{equation}
Thus $\approx$ is finer than $\sim^c$. (i.e. $\sim^c$ is coarser than
$\approx$, then the arbitrariness of $\approx$ implies that $\sim^c$ is
coarsest.)\qed
\end{proof}

As we know weak equivalence is not congruence, and one way to ``fix'' it,
is to use observation congruence which is based on weak equivalence
but have special treatments on the first transitions. The other way is to
build a congruence from existing weak equivalence relation, using
above approach based on one-hole contexts. Such a congruence has a new name:
\begin{definition}{(Weak bisimulation congruence)}
The coarasest congruence that is finer than weak bisimulation equivalence is called \emph{weak bisimulation
  congruence} (notation: $\sim_w^c$):
\begin{alltt}
\HOLTokenTurnstile{} \HOLConst{WEAK_CONGR} \HOLSymConst{=} (\HOLSymConst{=~})\HOLSymConst{\HOLTokenSupC{}}
or
\HOLTokenTurnstile{} \HOLConst{WEAK_CONGR} \HOLSymConst{=} (\HOLTokenLambda{}\HOLBoundVar{g} \HOLBoundVar{h}. \HOLSymConst{\HOLTokenForall{}}\HOLBoundVar{c}. \HOLConst{CONTEXT} \HOLBoundVar{c} \HOLSymConst{\HOLTokenImp{}} \HOLBoundVar{c} \HOLBoundVar{g} \HOLSymConst{\ensuremath{\approx}} \HOLBoundVar{c} \HOLBoundVar{h})
\end{alltt}
\end{definition}

So far, the weak bisimulation congruence $\sim_w^c$ defined above is irrelevant with 
rooted weak bisimulation (a.k.a. observation congruence) $\approx^c$,
which has the following standard definition also based on
weak equivalence:
\begin{alltt}
\HOLTokenTurnstile{} \HOLFreeVar{E} \HOLSymConst{\HOLTokenObsCongr} \HOLFreeVar{E\sp{\prime}} \HOLSymConst{\HOLTokenEquiv{}}
   \HOLSymConst{\HOLTokenForall{}}\HOLBoundVar{u}.
     (\HOLSymConst{\HOLTokenForall{}}\HOLBoundVar{E\sb{\mathrm{1}}}. \HOLFreeVar{E} --\HOLBoundVar{u}-> \HOLBoundVar{E\sb{\mathrm{1}}} \HOLSymConst{\HOLTokenImp{}} \HOLSymConst{\HOLTokenExists{}}\HOLBoundVar{E\sb{\mathrm{2}}}. \HOLFreeVar{E\sp{\prime}} ==\HOLBoundVar{u}=>> \HOLBoundVar{E\sb{\mathrm{2}}} \HOLSymConst{\HOLTokenConj{}} \HOLBoundVar{E\sb{\mathrm{1}}} \HOLSymConst{\ensuremath{\approx}} \HOLBoundVar{E\sb{\mathrm{2}}}) \HOLSymConst{\HOLTokenConj{}}
     \HOLSymConst{\HOLTokenForall{}}\HOLBoundVar{E\sb{\mathrm{2}}}. \HOLFreeVar{E\sp{\prime}} --\HOLBoundVar{u}-> \HOLBoundVar{E\sb{\mathrm{2}}} \HOLSymConst{\HOLTokenImp{}} \HOLSymConst{\HOLTokenExists{}}\HOLBoundVar{E\sb{\mathrm{1}}}. \HOLFreeVar{E} ==\HOLBoundVar{u}=>> \HOLBoundVar{E\sb{\mathrm{1}}} \HOLSymConst{\HOLTokenConj{}} \HOLBoundVar{E\sb{\mathrm{1}}} \HOLSymConst{\ensuremath{\approx}} \HOLBoundVar{E\sb{\mathrm{2}}}
\end{alltt}
But since obvervation congruence is congruence, it must be finer than
weak bisimulation congruence:
\begin{lemma}{(Obvervation congruence is finer than weak bisimulation
    congruence)}
\begin{alltt}
\HOLTokenTurnstile{} \HOLFreeVar{p} \HOLSymConst{\HOLTokenObsCongr} \HOLFreeVar{q} \HOLSymConst{\HOLTokenImp{}} \HOLConst{WEAK_CONGR} \HOLFreeVar{p} \HOLFreeVar{q}
\end{alltt}
\end{lemma}
On the other side, by consider the trivial context and sum contexts in the definition
of weak bisimulation congruence, we can easily prove the following
result:
\begin{lemma}
\begin{alltt}
\HOLTokenTurnstile{} \HOLConst{WEAK_CONGR} \HOLFreeVar{p} \HOLFreeVar{q} \HOLSymConst{\HOLTokenImp{}} \HOLConst{SUM_EQUIV} \HOLFreeVar{p} \HOLFreeVar{q}
\end{alltt}
\end{lemma}
Noticed that, in above theorem, the sum operator can be replaced by
any other operator in CCS, but we know sum is special because it's the
only operator in which the weak equivalence is not preserved after
substitutions.

From above two lemmas, we can easily see that, weak
equivalence is between the observation congruence and an unnamed
relation $\{ (p, q) \colon \forall\,r. p + r \approx q + r \}$ (we can
temporarily call it ``sum equivalence'', because we don't if it's a
congruence, or even if it's contained in weak equivalence). If we
could further prove that ``sum equivalence'' is finer than observation congruence:
\begin{proposition}
\begin{alltt}
\HOLTokenTurnstile{} (\HOLSymConst{\HOLTokenForall{}}\HOLBoundVar{r}. \HOLFreeVar{p} \HOLSymConst{+} \HOLBoundVar{r} \HOLSymConst{\ensuremath{\approx}} \HOLFreeVar{q} \HOLSymConst{+} \HOLBoundVar{r}) \HOLSymConst{\HOLTokenImp{}} \HOLFreeVar{p} \HOLSymConst{\HOLTokenObsCongr} \HOLFreeVar{q}
\end{alltt}
\end{proposition}
then all three congruences (observation congruence, weak equivalence
and the ``sum equivalence'' must all coincide, as illustrated in the
following figure:
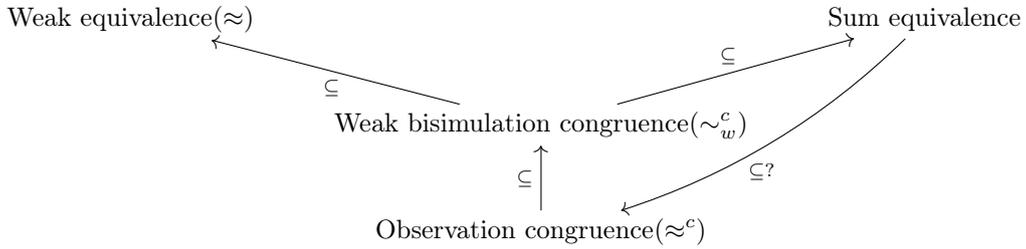
\begin{displaymath}
\xymatrix{
{\textrm{Weak equivalence} (\approx)} & {} & {\textrm{Sum
    equivalence}} \ar@/^3ex/[ldd]^{\subseteq ?}\\
{} & {\textrm{Weak bisimulation congruence} (\sim_w^c)}
\ar[lu]^{\subseteq} \ar[ru]^{\subseteq} \\
{} & {\textrm{Observation congruence} (\approx^c)} \ar[u]^{\subseteq}
}
\end{displaymath}

This is why the proposition at the begining of this section is called ``coarsest congruence
contained in weak equivalence'', it's actually trying to prove the
``sum equivalence'' is finer than ``observation congruence'' therefore
makes ``weak bisimulation congruence'' ($\sim_w^c$) coincide with ``observation
congruence'' ($\approx^c$).

\section{Coarsest congruence contained in weak equivalence}

The hightlight of this project is the formal proofs of various
versions of the so-called ``coarsest congruence contained in
$\approx$'' theorem:
\begin{proposition}
\begin{alltt}
\HOLTokenTurnstile{} \HOLFreeVar{p} \HOLSymConst{\HOLTokenObsCongr} \HOLFreeVar{q} \HOLSymConst{\HOLTokenEquiv{}} \HOLSymConst{\HOLTokenForall{}}\HOLBoundVar{r}. \HOLFreeVar{p} \HOLSymConst{+} \HOLBoundVar{r} \HOLSymConst{\ensuremath{\approx}} \HOLFreeVar{q} \HOLSymConst{+} \HOLBoundVar{r}
\end{alltt}
\end{proposition}
It's surprising hard  to prove this result, when there's no
assumptions on the processes.  We consider the following three cases with
increasing difficulities:
\begin{enumerate}
\item with classical cardinality assumptions;
\item for finite state CCS;
\item general case.
\end{enumerate}

The easy part (left $\Longrightarrow$ right) is already proven in
previous section by combining \texttt{OBS_CONGR_IMP_WEAK_CONGR} and
\texttt{WEAK_CONGR_IMP_SUM_EQUIV}, or it can be proved directly using
\texttt{OBS_CONGR_IMP_WEAK_EQUIV} and \texttt{OBS_CONGR_SUBST_SUM_R}:
\begin{theorem}(The easy part ``Coarsest congruence
  contained in $\approx$'')
\label{thm:easy-part}
\begin{alltt}
\HOLTokenTurnstile{} \HOLFreeVar{p} \HOLSymConst{\HOLTokenObsCongr} \HOLFreeVar{q} \HOLSymConst{\HOLTokenImp{}} \HOLSymConst{\HOLTokenForall{}}\HOLBoundVar{r}. \HOLFreeVar{p} \HOLSymConst{+} \HOLBoundVar{r} \HOLSymConst{\ensuremath{\approx}} \HOLFreeVar{q} \HOLSymConst{+} \HOLBoundVar{r} \hfill[COARSEST_CONGR_LR]
\end{alltt}
\end{theorem}
Thus we only focus on the hard part (right $\Longrightarrow$ left) in
the rest of this section.

\subsection{With classicial cardinality assumptions}

A classic restriction is to assume cardinality limitations on the two
processes, so that didn't use up all possible labels. Sometimes this
assumption is automatically satisfied, for example: the CCS is
finitrary and the set of all actions is infinite. But in our setting,
the CCS datatype contains twi type variables, and if the set of all possible
labels has only finite cardinalities, this assumtion may not be satisfied.

In \cite{Milner:2017tw} (Proposition 3 in Chapter 7, p. 153), Robin
Milner simply called
 this theorem the ``Proposition 3'':
\begin{proposition}{(Proposition 3 of observation congruence)}
Assume that $\mathcal{L}(P) \cup \mathcal{L}(Q) \neq
\mathcal{L}$. Then $P \approx^c Q$ iff, for all $R$, $P + R \approx Q
+ R$.
\end{proposition}
And in \cite{Gorrieri:2015jt} (Theorem 4.5 in Chapter 4, p. 185), Prof. Gorrieri
has called it ``Coarsest congruence contained in $\approx$'' (so did
us in this paper):
\begin{theorem}{(Coarsest congruence contained in $\approx$)}
Assume that $\mathrm{fn}(p) \cup \mathrm{fn}(q) \neq \mathscr{L}$. Then $p \approx^c q$
if and only if $p + r \approx q + r$ for all $r \in \mathscr{P}$.
\end{theorem}
Both $\mathcal{L}(\cdot)$ and $\mathrm{fn}(\cdot)$ used in above
theorems mean the set of ``non-$\tau$ actions'' (i.e. labels) used in a given process.

We analysized the proof of abvoe theorem and have
found that, the assumption that the two processes didn't use up all
available labels. Instead, it can be  weakened to the following
stronger version, which assumes the following properties instead:
\begin{definition}{(Processes having free actions)}
A CCS process is said to have \emph{free actions} if there exists an
non-$\tau$ action such that it doesn't appear in any transition or
weak transition directly leading from the root of the process:
\begin{alltt}
\HOLTokenTurnstile{} \HOLConst{free_action} \HOLFreeVar{p} \HOLSymConst{\HOLTokenEquiv{}} \HOLSymConst{\HOLTokenExists{}}\HOLBoundVar{a}. \HOLSymConst{\HOLTokenForall{}}\HOLBoundVar{p\sp{\prime}}. \HOLSymConst{\HOLTokenNeg{}}(\HOLFreeVar{p} ==\HOLConst{label} \HOLBoundVar{a}=>> \HOLBoundVar{p\sp{\prime}})
\end{alltt}
\end{definition}

\begin{theorem}{(Stronger version of ``Coarsest congruence contained
    in $\approx$'', only the hard part)}
Assuming for two processes $p$ and $q$ have free actions, then $p \approx^c q$ 
if $p + r \approx q + r$ for all $r \in \mathscr{P}$:
\begin{alltt}
\HOLTokenTurnstile{} \HOLConst{free_action} \HOLFreeVar{p} \HOLSymConst{\HOLTokenConj{}} \HOLConst{free_action} \HOLFreeVar{q} \HOLSymConst{\HOLTokenImp{}} (\HOLSymConst{\HOLTokenForall{}}\HOLBoundVar{r}. \HOLFreeVar{p} \HOLSymConst{+} \HOLBoundVar{r} \HOLSymConst{\ensuremath{\approx}} \HOLFreeVar{q} \HOLSymConst{+} \HOLBoundVar{r}) \HOLSymConst{\HOLTokenImp{}} \HOLFreeVar{p} \HOLSymConst{\HOLTokenObsCongr} \HOLFreeVar{q}
\end{alltt}
\end{theorem}
This new assumption is weaker because, even $p$ and $q$ may have used all possible
actions in their transition graphs, as long as there's one such free
action for their first-step weak transitions, therefore the theorem
still holds.
Also noticed that, the two processes do not have to share the same
free actions, this property focuses on single process.

\begin{proof}{(Proof of the stronger version of ``Coarsest congruence contained
    in $\approx$'')}
The kernel idea in this proof is to use that free action, say $a$, and
have $p + a.0 \approx q + a.0$ as the working basis. Then for any
transition from $p + a.0$, say $p + a.0 \overset{u}{\Longrightarrow}
E_1$, there must be a weak transition of the same action $u$ (or EPS
when $u = \tau$) coming
from $q + a.0$ as the response. We're going to use the free-action
assumptions to conclude that, when $u = \tau$, that EPS must contain
at least one $\tau$ (thus satisfied the definition of observation
congruence):
\begin{displaymath}
\xymatrix{
{p+a.0} \ar@{.}[r]^{\approx} \ar[d]^{u=\tau} & {q + a.0}
\ar@{=>}[d]^{\epsilon} \\
{E_1} \ar@{.}[r]^{\approx} & {E_2}
}
\end{displaymath}
Indeed, if the EPS leading from $q+a.0$ actually contains no
$\tau$-transition, that is, $q+a.0 = E_2$, then $E_1$ and $E_2$ cannot
be weak equivalence: for any $a$-transition from $q+a.0$, $E1$ must
response with a weak $a$-transition as $E_1
\overset{a}{\Longrightarrow} E_1'$, but this means $p
\overset{a}{\Longrightarrow} E_1'$, which is impossible by free-action
assumption on $p$:
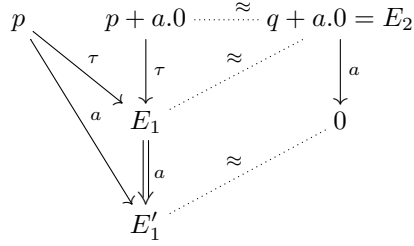
\begin{displaymath}
\xymatrix{
{p} \ar[rd]^{\tau} \ar[rdd]^{a} & {p + a.0} \ar@{.}[r]^{\approx} \ar[d]^{\tau}
 & {q+a.0 = E_2} \ar[d]^{a} \\
{} & {E_1} \ar@{.}[ru]^{\approx} \ar@{=>}[d]^{a} & {0} \\
{} & {E_1'} \ar@{.}[ru]^{\approx}
}
\end{displaymath}
Once we have $q + a.0 \overset{\tau}{\Longrightarrow} E2$, the first
$\tau$-transition must comes from $q$, then it's obvious to see that
$E_2$ is a valid response required by observation congruence of $p$
and $q$ in this case.

When $p \overset{L}{\longrightarrow} E_1$, we have $p+a.0
\overset{L}{\longrightarrow} E_1$, then there's an $E_2$ such that
$q+a.0\overset{L}{\Longrightarrow} E_2$. We can further conclude that
$q\overset{L}{\Longrightarrow}E_2$ because by free-action assumption
$L\neq a$. This finishes the first half of the proof, the second half
(for all transition coming from $q$) is completely symmetric. \qed
\end{proof}

Combining the easy and hard parts, the following theorem is proved:
\begin{theorem}{(Coarsest congruence contained in $\approx$)}
\begin{alltt}
\HOLTokenTurnstile{} \HOLConst{free_action} \HOLFreeVar{p} \HOLSymConst{\HOLTokenConj{}} \HOLConst{free_action} \HOLFreeVar{q} \HOLSymConst{\HOLTokenImp{}} (\HOLFreeVar{p} \HOLSymConst{\HOLTokenObsCongr} \HOLFreeVar{q} \HOLSymConst{\HOLTokenEquiv{}} \HOLSymConst{\HOLTokenForall{}}\HOLBoundVar{r}. \HOLFreeVar{p} \HOLSymConst{+} \HOLBoundVar{r} \HOLSymConst{\ensuremath{\approx}} \HOLFreeVar{q} \HOLSymConst{+} \HOLBoundVar{r})
\end{alltt}
\end{theorem}

\subsection{Without cardinality assumptions}

In 2005, Rob J. van Glabbeek published a paper
\cite{vanGlabbeek:2005ur} showing that ``the weak bisimulation
congruence can be characterised as rooted weak bisimulation
equivalence, even without making assumptions on the cardinality of the
sets of states or actions of the process under consideration''. That
is to say, above ``Coarsest congruence contained in $\approx$''
theorem holds even for two arbitrary processes! The idea is actually
from Jan Willem Klop back to the 80s, but it's not published until
that 2005 paper. This proof is not known to Robin Milner in \cite{Milner:2017tw}.
\footnote{We carefully
investigated this paper and focused on the formalization of the proof
contained in the paper, with all remain plans of this ``tirocinio''
project cancelled.}

The main result is the following version of the hard part of
``Coarsest congruence contained in $\approx$'' theorem under new
assumptions:
\begin{theorem}{(Coarsest congruence contained in $\approx$, new
    assumptions)}
\label{thm:new-assum}
For any two CCS processes $p$ and $q$, if there exists another stable
(i.e. first-step transitions are never $\tau$-transition) process
$k$ which is not weak bisimlar with any sub-process follows from $p$
and $q$ by \emph{one-step} weak transitions, then $p \approx^c q$
if $p + r \approx q + r$ for all $r \in \mathscr{P}$.
\begin{alltt}
\HOLTokenTurnstile{} (\HOLSymConst{\HOLTokenExists{}}\HOLBoundVar{k}.
      \HOLConst{STABLE} \HOLBoundVar{k} \HOLSymConst{\HOLTokenConj{}} (\HOLSymConst{\HOLTokenForall{}}\HOLBoundVar{p\sp{\prime}} \HOLBoundVar{u}. \HOLFreeVar{p} ==\HOLBoundVar{u}=>> \HOLBoundVar{p\sp{\prime}} \HOLSymConst{\HOLTokenImp{}} \HOLSymConst{\HOLTokenNeg{}}(\HOLBoundVar{p\sp{\prime}} \HOLSymConst{\ensuremath{\approx}} \HOLBoundVar{k})) \HOLSymConst{\HOLTokenConj{}}
      \HOLSymConst{\HOLTokenForall{}}\HOLBoundVar{q\sp{\prime}} \HOLBoundVar{u}. \HOLFreeVar{q} ==\HOLBoundVar{u}=>> \HOLBoundVar{q\sp{\prime}} \HOLSymConst{\HOLTokenImp{}} \HOLSymConst{\HOLTokenNeg{}}(\HOLBoundVar{q\sp{\prime}} \HOLSymConst{\ensuremath{\approx}} \HOLBoundVar{k})) \HOLSymConst{\HOLTokenImp{}}
   (\HOLSymConst{\HOLTokenForall{}}\HOLBoundVar{r}. \HOLFreeVar{p} \HOLSymConst{+} \HOLBoundVar{r} \HOLSymConst{\ensuremath{\approx}} \HOLFreeVar{q} \HOLSymConst{+} \HOLBoundVar{r}) \HOLSymConst{\HOLTokenImp{}}
   \HOLFreeVar{p} \HOLSymConst{\HOLTokenObsCongr} \HOLFreeVar{q}
\end{alltt}
\end{theorem}

\begin{proof}
Assuming the existence of that special process $k$, and take an
arbitrary non-$\tau$ action, say $a$ (this is always possible in our
setting, because in higher order logic any valid type must contain at
least one value), we'll use the fact that $p + a.k \approx q + a.k$ as
our working basis. For all transitions from $p$, say $p
\overset{u}{\longrightarrow} E_1$, we're going to prove that, there
must be a corresponding weak transition such that $q
\overset{u}{\Longrightarrow} E_2$, and $E_1 \approx E_2$ (thus $p
\approx^c q$. There're three cases to consider:

\begin{enumerate}
\item $\tau$-transitions: $p
\overset{\tau}{\longrightarrow} E_1$. By SOS rule ($\mathrm{Sum}_1$),
we have $p + a.k \overset{\tau}{\longrightarrow} E_1$, now by $p + a.k
\approx q + a.k$ and the property (*) of weak equivalence, there
exists an $E_2$ such that $q + a.k \overset{\epsilon}{\Longrightarrow}
E_2$. We can
use the property of $k$ to assert that, such an EPS transition must
contains at least one $\tau$-transition. Because if it's not, then $q
+a.k = E_2$, and since $E_1 \approx E_2$, for transition $q + a.k
\overset{a}{\longrightarrow} k$, $E_1$ must make a response by $E_1
\overset{a}{\Longrightarrow} E_1'$, and as the result we have $p
\overset{a}{\Longrightarrow} E_1'$ and $E_1' \approx k$, which is
impossible by the special choice of $k$:
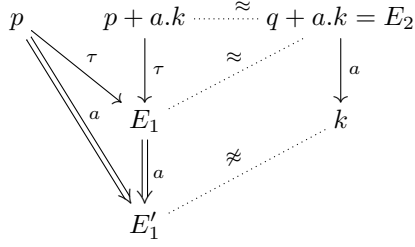
\begin{displaymath}
\xymatrix{
{p} \ar[dr]^{\tau} \ar@{=>}[ddr]^{a} & {p+a.k} \ar@{.}[r]^{\approx}
\ar[d]^{\tau} & {q+a.k = E_2} \ar[d]^{a} \\
{} & {E_1} \ar@{.}[ru]^{\approx} \ar@{=>}[d]^{a} & {k} \\
{} & {E_1'} \ar@{.}[ru]^{\not\approx}
}
\end{displaymath}
\item If there's a $a$-transition coming from $p$ (means that the
  arbitrary chosen action $a$ is normally used by processes $p$ and
  $q$), that is, $p \overset{a}{\longrightarrow} E_1$, also
  $p+a.k\overset{a}{\longrightarrow} E_1$,
by property (*) of weak equivalence, there exists $E_2$ such that $q +
a.k \overset{a}{\Longrightarrow} E_2$:
\begin{displaymath}
\xymatrix{
{p} \ar[dr]^{a} & {p+a.k} \ar@{.}[r]^{\approx} \ar[d]^{a} & {q+a.k}
\ar@{=>}[d]^{a} \\
{} & {\forall E_1} \ar@{.}[r]^{\approx} & {\exists E_2}
}
\end{displaymath}
We must further divide this weak transition into two cases based on its first step:
\begin{enumerate}
\item If the first step is a $\tau$-transition, then for sure this
  entire weak transition must come from $q$ (otherwise the first step
  would be an $a$-transition from $a.k$). And in this case we can
  easily conclude $q \overset{a}{\Longrightarrow} E_2$ without using
 the property of $k$:
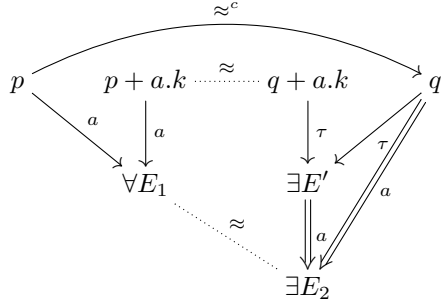
\begin{displaymath}
\xymatrix{
{p} \ar[dr]^{a} \ar@/^5ex/[rrr]^{\approx^c} & {p+a.k} \ar@{.}[r]^{\approx} \ar[d]^{a} & {q+a.k}
\ar[d]^{\tau} & {q} \ar[ld]^{\tau} \ar@{=>}[ldd]^{a} \\
{} & {\forall E_1} \ar@{.}[rd]^{\approx} & {\exists E'} \ar@{=>}[d]^{a} \\
{} & {} & {\exists E_2}
}
\end{displaymath}
\item If the first step is an $a$-transition, we can prove that, this
  $a$-transition must come from $h$ (then the proof finishes for the
  entire $a$-transition case). Because if it's from the $a.k$, since
  $k$ is stable, then there's no other coice but $E_2 = k$ and $E_1
  \approx E_2$. This is again impossible for the special choice of
  $k$:
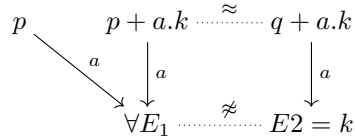
\begin{displaymath}
\xymatrix{
{p} \ar[dr]^{a} & {p+a.k} \ar@{.}[r]^{\approx} \ar[d]^{a} & {q+a.k}
\ar[d]^{a} \\
{} & {\forall E_1} \ar@{.}[r]^{\not\approx} & {E2 = k}
}
\end{displaymath}
\end{enumerate}

\item For other $L$-transitions coming from $p$, where $L \neq a$ and
  $L \neq \tau$. As a response to $p + a.k \overset{L}{\longrightarrow}
  E_1$, we have $q + a.k \overset{L}{\Longrightarrow} E_2$ and $E_1
  \approx E_2$. It's obvious that $q \overset{L}{\Longrightarrow} E_2$
  in this case, no matter what the first step is (it can only be $\tau$
  and $L$) and this satisfies the requirement of observation
  congruence natually:
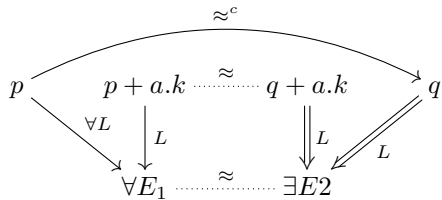
\begin{displaymath}
\xymatrix{
{p} \ar[dr]^{\forall L} \ar@/^5ex/[rrr]^{\approx^c} & {p+a.k} \ar@{.}[r]^{\approx} \ar[d]^{L} & {q+a.k}
\ar@{=>}[d]^{L} & {q} \ar@{=>}[ld]^{L} \\
{} & {\forall E_1} \ar@{.}[r]^{\approx} & {\exists E2}
}
\end{displaymath}
\end{enumerate}

The other direction (for all transitions coming from $q$) is
completely symmetric. Combining all the cases, we have $p \approx^c q$.\qed
\end{proof}
Now it remains to prove the existence of the special process mentioned in the assumption of above theorem.

\subsection{Arbitrary many non-bisimilar processes}

Strong equivalence, weak equivalence, observation congruence, they're
all equivalence relations on CCS process space. General speaking, each
equivalence relation must have \emph{partitioned} all processes into
several disjoint equivalence classes: processes in the same
equivalence class are equivalent, and processes in different
equivalence class are not equivalent.

The assumption in previous Theorem \ref{thm:new-assum} requires the
existence of a special CCS process, which is not weak equivalence to
any sub-process leading from the two root processes by weak
transitions. On worst cases, there may be infinite such
sub-processes\footnote{Even the CCS is finite branching, that's
  because after a weak transition, the end process may have an
  infinite $\tau$-chain, and with each $\tau$-transition added into
  the weak transition, the new end process is still a valid weak
  transition, thus lead to infinite number of weak transitions.} Thus
there's no essential differences to consider all states in the
process group instead.

Then it's natural to ask if there are infinite
equivalence classes of CCS processes. If so, then it should be 
possible to choose one which is not equivalent with all the (finite) states in the
graphs of the two given processes. It turns out that, after Jan Willem
Klop, it's possible to construct such processes, in which each of them forms a new
equivalence class, we call them ``Klop processes'' in this paper:
\begin{definition}{(Klop processes)}
For each ordinal $\lambda$, and an arbitrary chosen non-$\tau$ action $a$,
define a CCS process $k_\lambda$ as follows:
\begin{enumerate}
\item $k_0 = 0$,
\item $k_{\lambda+1} = k_\lambda + a.k_\lambda$ and
\item for $\lambda$ a limit ordinal, $k_\lambda = \sum_{\mu < \lambda}
  k_\mu$, meaning that $k_\lambda$ is constructed from all graphs
  $k_\mu$ for $\mu < \lambda$ by identifying their root.
\end{enumerate}
\end{definition}

Unfortunately, it's impossible to express infinite sums in our CCS
datatype settings\footnote{And such infinite sums seems to go beyond the
ability of the HOL's Datatype package} without intruducing new axioms.
Therefore we have followed a two-step approach in this project: first
we consider only the finite-state CCS (no need for axioms), then we turn
to the general case.

\subsection{Finite-state CCS}

If both processes $p$ and $q$ are finite-state CCS processes, that is,
the number of reachable states from $p$ and $q$ are both finite. And
in this case, the following limited version of Klop processes can be
defined as a recursive function (on natural numbers) in HOL4:
\begin{definition}{(Klop processes as recursive function on natural numbers)}
\begin{alltt}
\HOLConst{KLOP} \HOLFreeVar{a} \HOLNumLit{0} \HOLSymConst{=} \HOLConst{nil}
\HOLConst{KLOP} \HOLFreeVar{a} (\HOLConst{SUC} \HOLFreeVar{n}) \HOLSymConst{=} \HOLConst{KLOP} \HOLFreeVar{a} \HOLFreeVar{n} \HOLSymConst{+} \HOLConst{label} \HOLFreeVar{a}\HOLSymConst{..}\HOLConst{KLOP} \HOLFreeVar{a} \HOLFreeVar{n}\hfill[KLOP_def]
\end{alltt}
\end{definition}

By induction on the definition of Klop processes and SOS inference
rules ($\mathrm{Sum}_1$) and ($\mathrm{Sum}_2$), we can easily prove
the following properties of Klop functions:
\begin{proposition}{(Properties of Klop functions and processes)}
\begin{enumerate}
\item All Klop processes are stable:
\begin{alltt}
\HOLTokenTurnstile{} \HOLConst{STABLE} (\HOLConst{KLOP} \HOLFreeVar{a} \HOLFreeVar{n})\hfill[KLOP_PROP0]
\end{alltt}
\item All transitions of a Klop process must lead to another smaller Klop
  process, and any smaller Klop process must be a possible transition
  of a larger Klop process:
\begin{alltt}
\HOLTokenTurnstile{} \HOLConst{KLOP} \HOLFreeVar{a} \HOLFreeVar{n} --\HOLConst{label} \HOLFreeVar{a}-> \HOLFreeVar{E} \HOLSymConst{\HOLTokenEquiv{}} \HOLSymConst{\HOLTokenExists{}}\HOLBoundVar{m}. \HOLBoundVar{m} \HOLSymConst{\HOLTokenLt{}} \HOLFreeVar{n} \HOLSymConst{\HOLTokenConj{}} \HOLFreeVar{E} \HOLSymConst{=} \HOLConst{KLOP} \HOLFreeVar{a} \HOLBoundVar{m}\hfill[KLOP_PROP1]
\end{alltt}
\item The weak transition version of above property:
\begin{alltt}
\HOLTokenTurnstile{} \HOLConst{KLOP} \HOLFreeVar{a} \HOLFreeVar{n} ==\HOLConst{label} \HOLFreeVar{a}=>> \HOLFreeVar{E} \HOLSymConst{\HOLTokenEquiv{}} \HOLSymConst{\HOLTokenExists{}}\HOLBoundVar{m}. \HOLBoundVar{m} \HOLSymConst{\HOLTokenLt{}} \HOLFreeVar{n} \HOLSymConst{\HOLTokenConj{}} \HOLFreeVar{E} \HOLSymConst{=} \HOLConst{KLOP} \HOLFreeVar{a} \HOLBoundVar{m}\hfill[KLOP_PROP1']
\end{alltt}
\item All Klop processes are distinct according to strong equivalence:
\begin{alltt}
\HOLTokenTurnstile{} \HOLFreeVar{m} \HOLSymConst{\HOLTokenLt{}} \HOLFreeVar{n} \HOLSymConst{\HOLTokenImp{}} \HOLSymConst{\HOLTokenNeg{}}(\HOLConst{KLOP} \HOLFreeVar{a} \HOLFreeVar{m} \HOLSymConst{\ensuremath{\sim}} \HOLConst{KLOP} \HOLFreeVar{a} \HOLFreeVar{n})\hfill[KLOP_PROP2]
\end{alltt}
\item All Klop processes are distinct according to weak equivalence:
\begin{alltt}
\HOLTokenTurnstile{} \HOLFreeVar{m} \HOLSymConst{\HOLTokenLt{}} \HOLFreeVar{n} \HOLSymConst{\HOLTokenImp{}} \HOLSymConst{\HOLTokenNeg{}}(\HOLConst{KLOP} \HOLFreeVar{a} \HOLFreeVar{m} \HOLSymConst{\ensuremath{\approx}} \HOLConst{KLOP} \HOLFreeVar{a} \HOLFreeVar{n})\hfill[KLOP_PROP2']
\end{alltt}
\item Klop functions are one-one:
\begin{alltt}
\HOLTokenTurnstile{} \HOLConst{ONE_ONE} (\HOLConst{KLOP} \HOLFreeVar{a})\hfill{KLOP_ONE_ONE}
\end{alltt}
\end{enumerate}
\end{proposition}

Once we have a recursive function defined on all natural numbers $0, 1,
\ldots$, we can map them into a set containing all these Klop processes,
and the set is countable infinite. On the other side, the number of
all states coming from
two finite-state CCS processes $p$ and $q$ is finite. Choosing from an
infinite set for an element distinct with any subprocess leading from
$p$ and $q$, is always possible.  This result is purely mathematical,
completely falling into basic set theory:
\begin{lemma}
Given an equivalence relation $R$ defined on a type, and two sets $A, B$
of elements in this type, $A$ is finite, $B$ is infinite, and all elements
in $B$ are not equivalent, then there exists an element $k$ in $B$
which is not equivalent with any element in $A$:
\begin{alltt}
\HOLTokenTurnstile{} \HOLConst{equivalence} \HOLFreeVar{R} \HOLSymConst{\HOLTokenImp{}}
   \HOLConst{FINITE} \HOLFreeVar{A} \HOLSymConst{\HOLTokenConj{}} \HOLConst{INFINITE} \HOLFreeVar{B} \HOLSymConst{\HOLTokenConj{}}
   (\HOLSymConst{\HOLTokenForall{}}\HOLBoundVar{x} \HOLBoundVar{y}. \HOLBoundVar{x} \HOLSymConst{\HOLTokenIn{}} \HOLFreeVar{B} \HOLSymConst{\HOLTokenConj{}} \HOLBoundVar{y} \HOLSymConst{\HOLTokenIn{}} \HOLFreeVar{B} \HOLSymConst{\HOLTokenConj{}} \HOLBoundVar{x} \HOLSymConst{\HOLTokenNotEqual{}} \HOLBoundVar{y} \HOLSymConst{\HOLTokenImp{}} \HOLSymConst{\HOLTokenNeg{}}\HOLFreeVar{R} \HOLBoundVar{x} \HOLBoundVar{y}) \HOLSymConst{\HOLTokenImp{}}
   \HOLSymConst{\HOLTokenExists{}}\HOLBoundVar{k}. \HOLBoundVar{k} \HOLSymConst{\HOLTokenIn{}} \HOLFreeVar{B} \HOLSymConst{\HOLTokenConj{}} \HOLSymConst{\HOLTokenForall{}}\HOLBoundVar{n}. \HOLBoundVar{n} \HOLSymConst{\HOLTokenIn{}} \HOLFreeVar{A} \HOLSymConst{\HOLTokenImp{}} \HOLSymConst{\HOLTokenNeg{}}\HOLFreeVar{R} \HOLBoundVar{n} \HOLBoundVar{k}
\end{alltt}
\end{lemma}
\begin{proof}
  We built an explicit mapping $f$ from $A$ to $B$\footnote{There're
    multiple ways to prove this lemma, a simpler proof is to make a
    reverse mapping from $B$ to the power set of $A$ (or further use
    the Axiom of Choice (AC) to make a mapping from $B$ to $A$), then
    the non-injectivity of this mapping will contradict the fact that
    all elements in the infinite set are distinct. Our proof doesn't
    need AC, and it relies on very simple truths about sets.}, for all
  $x \in A$, $y = f(x)$ if $y \in B$ and $y$ is equivalent with
  $x$. But it's possible that no element in $B$ is equivalent with
  $x$, and in this case we just choose an arbitrary element as
  $f(x)$. Such a mapping is to make sure the range of $f$ always fall
  into $B$.

  Now we can map $A$ to a subset of $B$, say $B_0$, and the
  cardinality of $B_0$ must be equal or smaller than the cardinality
  of $A$, thus finite. Now we choose an element $k$ from the rest part
  of $B$, this element is the desire one, because for any element
  $x \in A$, if it's equivalent with $k$, consider two cases for
  $y = f(x) \in B_0$:
  \begin{enumerate}
  \item $y$ is equivalent with $x$. In this case by transitivity of
    $R$, we have two distinct elements $y$ and $k$, one in $B_0$, the
    other in $B\setminus B_0$, they're equivalent. This violates the
    assumption that all elements in $B$ are distinct.
  \item $y$ is arbitrary chosen because there's no equivalent element
    for $x$ in $B$. But we already know one: $k$.
  \end{enumerate}
  Thus there's no element $x$ (in $A$) which is equivalent with $k$.\qed
\end{proof}

To reason about finite-state CCS, we also need to define the concept
of ``finite-state'':
\begin{definition}{(Definitions related to finite-state CCS)}
\begin{enumerate}
\item Define \emph{reachable} as the RTC of a relation, which
  indicates the existence of a transition between two processes:
\begin{alltt}
\HOLTokenTurnstile{} \HOLConst{Reachable} \HOLSymConst{=} (\HOLTokenLambda{}\HOLBoundVar{E} \HOLBoundVar{E\sp{\prime}}. \HOLSymConst{\HOLTokenExists{}}\HOLBoundVar{u}. \HOLBoundVar{E} --\HOLBoundVar{u}-> \HOLBoundVar{E\sp{\prime}})\HOLSymConst{\HOLTokenSupStar{}}
\end{alltt}
\item The ``nodes'' of a process is the set of all processes reachable
  from it:
\begin{alltt}
\HOLTokenTurnstile{} \HOLConst{NODES} \HOLFreeVar{p} \HOLSymConst{=} \HOLTokenLeftbrace{}\HOLBoundVar{q} \HOLTokenBar{} \HOLConst{Reachable} \HOLFreeVar{p} \HOLBoundVar{q}\HOLTokenRightbrace{}
\end{alltt}
\item A process is finite-state if the set of nodes is finite:
\begin{alltt}
\HOLTokenTurnstile{} \HOLConst{FINITE_STATE} \HOLFreeVar{p} \HOLSymConst{\HOLTokenEquiv{}} \HOLConst{FINITE} (\HOLConst{NODES} \HOLFreeVar{p})
\end{alltt}
\end{enumerate}
\end{definition}
Among many properties of above definitions, we mainly rely on the
following ``obvious'' property on weak transitions:
\begin{proposition}
If $p$ weakly transit to $q$, then $q$ must be in the node set of $p$:
\begin{alltt}
\HOLTokenTurnstile{} \HOLFreeVar{p} ==\HOLFreeVar{u}=>> \HOLFreeVar{q} \HOLSymConst{\HOLTokenImp{}} \HOLFreeVar{q} \HOLSymConst{\HOLTokenIn{}} \HOLConst{NODES} \HOLFreeVar{p}\hfill[WEAK_TRANS_IN_NODES]
\end{alltt}
\end{proposition}

Using all above results, now we can easily prove the following finite
version of ``Klop lemma'':
\begin{lemma}{Klop lemma, the finite version}
\label{lem:klop-lemma-finite}
For any two finite-state CCS $p$ and $q$, there exists another process $k$, which
is not weak equivalent with any sub-process weakly transited from $p$
and $q$:
\begin{alltt}
\HOLTokenTurnstile{} \HOLSymConst{\HOLTokenForall{}}\HOLBoundVar{p} \HOLBoundVar{q}.
     \HOLConst{FINITE_STATE} \HOLBoundVar{p} \HOLSymConst{\HOLTokenConj{}} \HOLConst{FINITE_STATE} \HOLBoundVar{q} \HOLSymConst{\HOLTokenImp{}}
     \HOLSymConst{\HOLTokenExists{}}\HOLBoundVar{k}.
       \HOLConst{STABLE} \HOLBoundVar{k} \HOLSymConst{\HOLTokenConj{}} (\HOLSymConst{\HOLTokenForall{}}\HOLBoundVar{p\sp{\prime}} \HOLBoundVar{u}. \HOLBoundVar{p} ==\HOLBoundVar{u}=>> \HOLBoundVar{p\sp{\prime}} \HOLSymConst{\HOLTokenImp{}} \HOLSymConst{\HOLTokenNeg{}}(\HOLBoundVar{p\sp{\prime}} \HOLSymConst{\ensuremath{\approx}} \HOLBoundVar{k})) \HOLSymConst{\HOLTokenConj{}}
       \HOLSymConst{\HOLTokenForall{}}\HOLBoundVar{q\sp{\prime}} \HOLBoundVar{u}. \HOLBoundVar{q} ==\HOLBoundVar{u}=>> \HOLBoundVar{q\sp{\prime}} \HOLSymConst{\HOLTokenImp{}} \HOLSymConst{\HOLTokenNeg{}}(\HOLBoundVar{q\sp{\prime}} \HOLSymConst{\ensuremath{\approx}} \HOLBoundVar{k})\hfill[KLOP_LEMMA_FINITE]
\end{alltt}
\end{lemma}

Combining above lemma, Theorem \ref{thm:new-assum} and Theorem
\ref{thm:easy-part}, we can easily prove the following theorem for finite-state CCS:
\begin{theorem}{(Coarsest congruence contained in $\approx$ for
    finite-state CCS)}
\begin{alltt}
\HOLTokenTurnstile{} \HOLConst{FINITE_STATE} \HOLFreeVar{p} \HOLSymConst{\HOLTokenConj{}} \HOLConst{FINITE_STATE} \HOLFreeVar{q} \HOLSymConst{\HOLTokenImp{}}
   (\HOLFreeVar{p} \HOLSymConst{\HOLTokenObsCongr} \HOLFreeVar{q} \HOLSymConst{\HOLTokenEquiv{}} \HOLSymConst{\HOLTokenForall{}}\HOLBoundVar{r}. \HOLFreeVar{p} \HOLSymConst{+} \HOLBoundVar{r} \HOLSymConst{\ensuremath{\approx}} \HOLFreeVar{q} \HOLSymConst{+} \HOLBoundVar{r})
\end{alltt}
\end{theorem}

\subsection{General case}

Now we turn to the general case.\footnote{This part in the project is
  not commited into HOL official repository, because of the possibly
  wrong use of \texttt{new_axiom}, full project code can be found at \url{https://github.com/binghe/informatica-public/tree/master/CCS2}}. The number of nodes in the graph of a
CCS process may be infinite, and in worst case such an ``infinite''
may be uncountable or even larger. In such cases, it's not guaranteed to
find a Klop process $K_n, n \in \mathbb{N}$ which is not weak
equivalence with any node (sub-process) in the graph. To formalize
such a proof, we have to use ordinals instead of natural numbers in the
definition of Klop processes.

Unfortunately, due to limitations in higher order logic, the CCS datatype has no way to express infinite sum
of CCS processes, e.g. an constructor
\HOLinline{(\HOLConst{summ} :((\ensuremath{\alpha}, \ensuremath{\beta}) \HOLTyOp{CCS} -> \HOLTyOp{bool}) -> (\ensuremath{\alpha}, \ensuremath{\beta}) \HOLTyOp{CCS})}.\footnote{Michael Norrish, HOL maintainer,
  explains the reason: ``You can’t define a type that recurses under
  the set `constructor' (your summ constructor has (CCS set) as an
  argument).  Ignoring the num set argument, you would then have an
  injective function (the summ constructor itself) from sets of CCS
  values into single CCS values. This ultimately falls foul of
  Cantor’s proof that the power set is strictly larger than the set.''
Michael further asserts that, in theory it's possible to have an constructor of type ``\HOLinline{\HOLTyOp{num} -> (\ensuremath{\alpha}, \ensuremath{\beta}) \HOLTyOp{CCS}}'', in which the sub-type \HOLinline{\HOLTyOp{num}} can be replaced to
\HOLinline{\ensuremath{\gamma} \HOLTyOp{ordinal}} to support injection from ordinals to a set of CCS
processes, then the type variable of ordinals becomes part of the CCS
datatype, e.g. \texttt{('a, 'b, 'c) CCS}. However, due to limitations
in HOL's ordinal theory, we can't further use the CCS type to prove
the existence of an ordinal which is larger than the cardinality of
the set of all CCS processes.}  As a result, we had no choice but to
introduce a new ``axiom'' for reasoning about infinite sums of CCS:
\begin{proposition}{(Infinite sum axiom for CCS)}
\begin{alltt}
\HOLTokenTurnstile{} \HOLSymConst{\HOLTokenExists{}}\HOLBoundVar{f}. \HOLSymConst{\HOLTokenForall{}}\HOLBoundVar{rs} \HOLBoundVar{u} \HOLBoundVar{E}. \HOLBoundVar{f} \HOLBoundVar{rs} --\HOLBoundVar{u}-> \HOLBoundVar{E} \HOLSymConst{\HOLTokenEquiv{}} \HOLSymConst{\HOLTokenExists{}}\HOLBoundVar{r}. \HOLBoundVar{r} \HOLSymConst{\HOLTokenIn{}} \HOLBoundVar{rs} \HOLSymConst{\HOLTokenConj{}} \HOLBoundVar{r} --\HOLBoundVar{u}-> \HOLBoundVar{E}
\end{alltt}
\end{proposition}
Above axiom simply asserts the existence of an infinite sum of CCS
processes, and all its transitions come from the transition of any
process in the set. With above axiom, we can then define the infinite
summ operator (through its behavior):
\begin{definition}
\begin{alltt}
\HOLTokenTurnstile{} \HOLConst{summ} \HOLFreeVar{rs} --\HOLFreeVar{u}-> \HOLFreeVar{E} \HOLSymConst{\HOLTokenEquiv{}} \HOLSymConst{\HOLTokenExists{}}\HOLBoundVar{r}. \HOLBoundVar{r} \HOLSymConst{\HOLTokenIn{}} \HOLFreeVar{rs} \HOLSymConst{\HOLTokenConj{}} \HOLBoundVar{r} --\HOLFreeVar{u}-> \HOLFreeVar{E}
\end{alltt}
\end{definition}

Now we can define the full version of Klop function based on ordinals:
\begin{definition}{(Klop function in HOL4, full version)}
\begin{quote}
\begin{alltt}
\HOLConst{Klop} \HOLFreeVar{a} \HOLNumLit{0o} \HOLSymConst{=} \HOLConst{nil}
\HOLConst{Klop} \HOLFreeVar{a} \HOLFreeVar{\ensuremath{\alpha}}\HOLSymConst{\HOLTokenSupPlus{}} \HOLSymConst{=} \HOLConst{Klop} \HOLFreeVar{a} \HOLFreeVar{\ensuremath{\alpha}} \HOLSymConst{+} \HOLConst{label} \HOLFreeVar{a}\HOLSymConst{..}\HOLConst{Klop} \HOLFreeVar{a} \HOLFreeVar{\ensuremath{\alpha}}
\HOLNumLit{0o} \HOLSymConst{\HOLTokenLt{}} \HOLFreeVar{\ensuremath{\alpha}} \HOLSymConst{\HOLTokenConj{}} \HOLConst{islimit} \HOLFreeVar{\ensuremath{\alpha}} \HOLSymConst{\HOLTokenImp{}}
\HOLConst{Klop} \HOLFreeVar{a} \HOLFreeVar{\ensuremath{\alpha}} \HOLSymConst{=} \HOLConst{summ} (\HOLConst{IMAGE} (\HOLConst{Klop} \HOLFreeVar{a}) (\HOLConst{preds} \HOLFreeVar{\ensuremath{\alpha}}))
\end{alltt}
\end{quote}
\end{definition}
Using above definition, we can further prove the following ``cases''
theorem for possible transitions of infinite sums:
\begin{proposition}{(``cases'' theorems for transitions of Klop processes)}
\begin{alltt}
\HOLTokenTurnstile{} (\HOLSymConst{\HOLTokenForall{}}\HOLBoundVar{a}. \HOLConst{Klop} \HOLBoundVar{a} \HOLNumLit{0o} \HOLSymConst{=} \HOLConst{nil}) \HOLSymConst{\HOLTokenConj{}}
   (\HOLSymConst{\HOLTokenForall{}}\HOLBoundVar{a} \HOLBoundVar{n} \HOLBoundVar{u} \HOLBoundVar{E}.
      \HOLConst{Klop} \HOLBoundVar{a} \HOLBoundVar{n}\HOLSymConst{\HOLTokenSupPlus{}} --\HOLBoundVar{u}-> \HOLBoundVar{E} \HOLSymConst{\HOLTokenEquiv{}}
      \HOLBoundVar{u} \HOLSymConst{=} \HOLConst{label} \HOLBoundVar{a} \HOLSymConst{\HOLTokenConj{}} \HOLBoundVar{E} \HOLSymConst{=} \HOLConst{Klop} \HOLBoundVar{a} \HOLBoundVar{n} \HOLSymConst{\HOLTokenDisj{}} \HOLConst{Klop} \HOLBoundVar{a} \HOLBoundVar{n} --\HOLBoundVar{u}-> \HOLBoundVar{E}) \HOLSymConst{\HOLTokenConj{}}
   \HOLSymConst{\HOLTokenForall{}}\HOLBoundVar{a} \HOLBoundVar{n} \HOLBoundVar{u} \HOLBoundVar{E}.
     \HOLNumLit{0o} \HOLSymConst{\HOLTokenLt{}} \HOLBoundVar{n} \HOLSymConst{\HOLTokenConj{}} \HOLConst{islimit} \HOLBoundVar{n} \HOLSymConst{\HOLTokenImp{}}
     (\HOLConst{Klop} \HOLBoundVar{a} \HOLBoundVar{n} --\HOLBoundVar{u}-> \HOLBoundVar{E} \HOLSymConst{\HOLTokenEquiv{}} \HOLSymConst{\HOLTokenExists{}}\HOLBoundVar{m}. \HOLBoundVar{m} \HOLSymConst{\HOLTokenLt{}} \HOLBoundVar{n} \HOLSymConst{\HOLTokenConj{}} \HOLConst{Klop} \HOLBoundVar{a} \HOLBoundVar{m} --\HOLBoundVar{u}-> \HOLBoundVar{E})
\end{alltt}
\end{proposition}
We can also converted them into the following inference rules for transitions
which are easier for use:\footnote{But these rules alone did not completely capture all the
behaviors of Klop processes, because they only talked about the valid
transitions and said nothing about invalid transitions}
\begin{proposition}{(``rules'' theorems for transitions of Klop processes)}
\begin{alltt}
\HOLTokenTurnstile{} (\HOLSymConst{\HOLTokenForall{}}\HOLBoundVar{a} \HOLBoundVar{n}. \HOLConst{Klop} \HOLBoundVar{a} \HOLBoundVar{n}\HOLSymConst{\HOLTokenSupPlus{}} --\HOLConst{label} \HOLBoundVar{a}-> \HOLConst{Klop} \HOLBoundVar{a} \HOLBoundVar{n}) \HOLSymConst{\HOLTokenConj{}}
   \HOLSymConst{\HOLTokenForall{}}\HOLBoundVar{a} \HOLBoundVar{n} \HOLBoundVar{m} \HOLBoundVar{u} \HOLBoundVar{E}.
     \HOLNumLit{0o} \HOLSymConst{\HOLTokenLt{}} \HOLBoundVar{n} \HOLSymConst{\HOLTokenConj{}} \HOLConst{islimit} \HOLBoundVar{n} \HOLSymConst{\HOLTokenConj{}} \HOLBoundVar{m} \HOLSymConst{\HOLTokenLt{}} \HOLBoundVar{n} \HOLSymConst{\HOLTokenConj{}} \HOLConst{Klop} \HOLBoundVar{a} \HOLBoundVar{m} --\HOLBoundVar{u}-> \HOLBoundVar{E} \HOLSymConst{\HOLTokenImp{}}
     \HOLConst{Klop} \HOLBoundVar{a} \HOLBoundVar{n} --\HOLBoundVar{u}-> \HOLBoundVar{E}
\end{alltt}
\end{proposition}

Using transfinite induction, we can prove the following properties of
the new Klop processes based on ordinals, which is the same with the
finite version of Klop processes:
\begin{alltt}
\HOLTokenTurnstile{} \HOLConst{STABLE} (\HOLConst{Klop} \HOLFreeVar{a} \HOLFreeVar{n})\hfill[Klop_PROP0]
\HOLTokenTurnstile{} \HOLConst{Klop} \HOLFreeVar{a} \HOLFreeVar{n} --\HOLConst{label} \HOLFreeVar{a}-> \HOLFreeVar{E} \HOLSymConst{\HOLTokenEquiv{}} \HOLSymConst{\HOLTokenExists{}}\HOLBoundVar{m}. \HOLBoundVar{m} \HOLSymConst{\HOLTokenLt{}} \HOLFreeVar{n} \HOLSymConst{\HOLTokenConj{}} \HOLFreeVar{E} \HOLSymConst{=} \HOLConst{Klop} \HOLFreeVar{a} \HOLBoundVar{m}\hfill[Klop_PROP1]
\HOLTokenTurnstile{} \HOLConst{Klop} \HOLFreeVar{a} \HOLFreeVar{n} ==\HOLConst{label} \HOLFreeVar{a}=>> \HOLFreeVar{E} \HOLSymConst{\HOLTokenEquiv{}} \HOLSymConst{\HOLTokenExists{}}\HOLBoundVar{m}. \HOLBoundVar{m} \HOLSymConst{\HOLTokenLt{}} \HOLFreeVar{n} \HOLSymConst{\HOLTokenConj{}} \HOLFreeVar{E} \HOLSymConst{=} \HOLConst{Klop} \HOLFreeVar{a} \HOLBoundVar{m}\hfill[Klop_PROP1']
\HOLTokenTurnstile{} \HOLFreeVar{m} \HOLSymConst{\HOLTokenLt{}} \HOLFreeVar{n} \HOLSymConst{\HOLTokenImp{}} \HOLSymConst{\HOLTokenNeg{}}(\HOLConst{Klop} \HOLFreeVar{a} \HOLFreeVar{m} \HOLSymConst{\ensuremath{\sim}} \HOLConst{Klop} \HOLFreeVar{a} \HOLFreeVar{n})\hfill[Klop_PROP2]
\HOLTokenTurnstile{} \HOLFreeVar{m} \HOLSymConst{\HOLTokenLt{}} \HOLFreeVar{n} \HOLSymConst{\HOLTokenImp{}} \HOLSymConst{\HOLTokenNeg{}}(\HOLConst{Klop} \HOLFreeVar{a} \HOLFreeVar{m} \HOLSymConst{\ensuremath{\approx}} \HOLConst{Klop} \HOLFreeVar{a} \HOLFreeVar{n})\hfill[Klop_PROP2']
\HOLTokenTurnstile{} \HOLConst{ONE_ONE} (\HOLConst{Klop} \HOLFreeVar{a})\hfill[Klop_ONE_ONE]
\end{alltt}
The transfinite induction principles we have used here, is the following
two theorems in HOL's \texttt{ordinalTheory}:
\begin{alltt}
\HOLTokenTurnstile{} (\HOLSymConst{\HOLTokenForall{}}\HOLBoundVar{min}. (\HOLSymConst{\HOLTokenForall{}}\HOLBoundVar{b}. \HOLBoundVar{b} \HOLSymConst{\HOLTokenLt{}} \HOLBoundVar{min} \HOLSymConst{\HOLTokenImp{}} \HOLFreeVar{P} \HOLBoundVar{b}) \HOLSymConst{\HOLTokenImp{}} \HOLFreeVar{P} \HOLBoundVar{min}) \HOLSymConst{\HOLTokenImp{}} \HOLSymConst{\HOLTokenForall{}}\HOLBoundVar{\ensuremath{\alpha}}. \HOLFreeVar{P} \HOLBoundVar{\ensuremath{\alpha}}\hfill[ord_inductition]
\HOLTokenTurnstile{} \HOLFreeVar{P} \HOLNumLit{0} \HOLSymConst{\HOLTokenConj{}} (\HOLSymConst{\HOLTokenForall{}}\HOLBoundVar{\ensuremath{\alpha}}. \HOLFreeVar{P} \HOLBoundVar{\ensuremath{\alpha}} \HOLSymConst{\HOLTokenImp{}} \HOLFreeVar{P} \HOLBoundVar{\ensuremath{\alpha}}\HOLSymConst{\HOLTokenSupPlus{}}) \HOLSymConst{\HOLTokenConj{}}
   (\HOLSymConst{\HOLTokenForall{}}\HOLBoundVar{\ensuremath{\alpha}}. \HOLConst{islimit} \HOLBoundVar{\ensuremath{\alpha}} \HOLSymConst{\HOLTokenConj{}} \HOLNumLit{0} \HOLSymConst{\HOLTokenLt{}} \HOLBoundVar{\ensuremath{\alpha}} \HOLSymConst{\HOLTokenConj{}} (\HOLSymConst{\HOLTokenForall{}}\HOLBoundVar{\ensuremath{\beta}}. \HOLBoundVar{\ensuremath{\beta}} \HOLSymConst{\HOLTokenLt{}} \HOLBoundVar{\ensuremath{\alpha}} \HOLSymConst{\HOLTokenImp{}} \HOLFreeVar{P} \HOLBoundVar{\ensuremath{\beta}}) \HOLSymConst{\HOLTokenImp{}} \HOLFreeVar{P} \HOLBoundVar{\ensuremath{\alpha}}) \HOLSymConst{\HOLTokenImp{}}
   \HOLSymConst{\HOLTokenForall{}}\HOLBoundVar{\ensuremath{\alpha}}. \HOLFreeVar{P} \HOLBoundVar{\ensuremath{\alpha}}\hfill[simple_ord_induction]
\end{alltt}
During the proofs of above properties, many basic results on ordinals
were also used, here we omit the proof details.

The next step is to prove the following important result:
\begin{theorem}
For any arbitrary set of CCS processes, it's always possible to
find a Klop process which is not weakly bisimilar with any process in
the set:
\begin{alltt}
\HOLTokenTurnstile{} \HOLSymConst{\HOLTokenForall{}}\HOLBoundVar{a} \HOLBoundVar{A}. \HOLSymConst{\HOLTokenExists{}}\HOLBoundVar{n}. \HOLSymConst{\HOLTokenForall{}}\HOLBoundVar{x}. \HOLBoundVar{x} \HOLSymConst{\HOLTokenIn{}} \HOLBoundVar{A} \HOLSymConst{\HOLTokenImp{}} \HOLSymConst{\HOLTokenNeg{}}(\HOLBoundVar{x} \HOLSymConst{\ensuremath{\approx}} \HOLConst{Klop} \HOLBoundVar{a} \HOLBoundVar{n})
\end{alltt}
\end{theorem}
\begin{proof}
Our formal proof depends on the following theorem in HOL's
\texttt{ordinalTheory}:
\begin{alltt}
\HOLTokenTurnstile{} \ensuremath{\cal{U}}(:\ensuremath{\alpha} \HOLTyOp{inf}) \HOLSymConst{\ensuremath{\prec}} \ensuremath{\cal{U}}(:\ensuremath{\alpha} \HOLTyOp{ordinal})\hfill[univ_ord_greater_cardinal]
\end{alltt}
which basically says the existence of ordinals larger than the
cadinality of any set, which is true in set theory. The HOL type
\HOLinline{\ensuremath{\alpha} \HOLTyOp{inf}} means the sum type of \HOLinline{\HOLTyOp{num}} and \HOLinline{\ensuremath{\alpha}}.

Here we must explain that, our formal proofs of mathematics theorems
is not based on \emph{Zermelo--Fraenkel (ZF)}
or \emph{von Neumann--Bernays--G\"{o}del (NBG)} set theory but a special
set theory in higher-order logic. It's know that, the typed logic
implemented in the various HOL systems (including
Isabelle/HOL) is not strong enough to define a type for all possible ordinal values (a
proper class in a set theory like NBG). Instead, there's a type
variable $\alpha$ in ordinals, and to apply above theorem, this type
variable must be connected with CCS datatype. Here is the sketch of our formal proof:

We define a mapping $f$ from ordinals to the union of natural numbers
and the \emph{power set} of $A$ which actually represents all CCS
processes in the graphs of two rooted processes $p$ and $q$:
\begin{equation}
f(n) = \begin{cases}
n & \text{if $n < \omega$,} \\
\{ y \colon y \in B \wedge y \approx Klop_n \} & \text{if $n \geqslant \omega$}.
\end{cases}
\end{equation}
Suppose the proposition is not true, that is, for each process $p$ in $A$,
there's at least one Klop process $k$ which is weakly bisimlar with
$p$. Then above mapping will never map any ordinal to empty set. And
the part for $n < \omega$ is obvious a bijection. And we know the rest
part of mapping is one-one.

Now the theorem \texttt{univ_ord_greater_cardinal} says there's no
injections from ordinals to set $A$, then there must be at least one non-empty
subset of $A$, and the process in it is weakly bisimilar with two
distinct Klop processes. By transitivity of weak equivalence, the two
Klop processes must also be weak equivalent, but this violates the
property 2 (weak version) of Klop processes.\qed
\end{proof}

A pure set-theory theorem sharing the same proof idea but with all
concurrency-theory stuff removed, is the following ``existence'' theorem:
\begin{theorem}
Assuming an arbitrary set $A$ of type $\alpha$, and a one-one mapping
$f$ from ordinals to type $\alpha$. There always exists an ordinal $n$
such that $f(n) \notin A$.
\begin{alltt}
\HOLTokenTurnstile{} \HOLSymConst{\HOLTokenForall{}}(\HOLBoundVar{A} :\ensuremath{\alpha} -> \HOLTyOp{bool}) (\HOLBoundVar{f} :\ensuremath{\alpha} \HOLTyOp{ordinal} -> \ensuremath{\alpha}).
     \HOLConst{ONE_ONE} \HOLBoundVar{f} \HOLSymConst{\HOLTokenImp{}} \HOLSymConst{\HOLTokenExists{}}(\HOLBoundVar{n} :\ensuremath{\alpha} \HOLTyOp{ordinal}). \HOLBoundVar{f} \HOLBoundVar{n} \HOLSymConst{\HOLTokenNotIn{}} \HOLBoundVar{A}
\end{alltt}
\end{theorem}
This result is elegant but unusual, because that ``arbitrary set''
can simply be the universe of all values of type $\alpha$, how can
there be another value (of the same type) not in it? Our answer is, in
such cases the mapping $f$ can't be one-one, and a false assumption
will lead to any conclusion in a theorem.\footnote{Above
theorem also indicates that, no matter how ``complicated'' a CCS process is, it's impossible
for it to contain all possible equivalence classes of CCS processes as
its sub-processes after certain transitions.}

Now we're ready to prove the following full version of ``Klop lemma'':
\begin{lemma}{(Klop lemma, the full version)}
\label{lem:klop-lemma-full}
For any two CCS processes $g$ and $h$, there exists another process
$k$ which is not weakly equivalent with any sub-process weakly transited from $g$
and $h$:
\begin{alltt}
\HOLTokenTurnstile{} \HOLSymConst{\HOLTokenForall{}}\HOLBoundVar{p} \HOLBoundVar{q}.
     \HOLSymConst{\HOLTokenExists{}}\HOLBoundVar{k}.
       \HOLConst{STABLE} \HOLBoundVar{k} \HOLSymConst{\HOLTokenConj{}} (\HOLSymConst{\HOLTokenForall{}}\HOLBoundVar{p\sp{\prime}} \HOLBoundVar{u}. \HOLBoundVar{p} ==\HOLBoundVar{u}=>> \HOLBoundVar{p\sp{\prime}} \HOLSymConst{\HOLTokenImp{}} \HOLSymConst{\HOLTokenNeg{}}(\HOLBoundVar{p\sp{\prime}} \HOLSymConst{\ensuremath{\approx}} \HOLBoundVar{k})) \HOLSymConst{\HOLTokenConj{}}
       \HOLSymConst{\HOLTokenForall{}}\HOLBoundVar{q\sp{\prime}} \HOLBoundVar{u}. \HOLBoundVar{q} ==\HOLBoundVar{u}=>> \HOLBoundVar{q\sp{\prime}} \HOLSymConst{\HOLTokenImp{}} \HOLSymConst{\HOLTokenNeg{}}(\HOLBoundVar{q\sp{\prime}} \HOLSymConst{\ensuremath{\approx}} \HOLBoundVar{k})\hfill[KLOP_LEMMA]
\end{alltt}
\end{lemma}
\begin{proof}
We consider the union \texttt{nodes} of all nodes (sub-processes) from $g$ and $h$. If the union is
finite, we use previous finite version of this lemma (and the finite
version of Klop processes which is well defined in HOL) to get the
conclusion.  If the union is infinite, we turn to use the full version
of Klop process defined (as axiom) on ordinals, and use the previous
theorems on ordinals to assert the existence of an ordinal $n$ such
that $Klop_n$ is not weakly bisimiar with any node in \texttt{nodes}. \qed
\end{proof}

And finally, with \emph{all above lemmas, theorems, definitions, plus
  one axiomatized definition of infinite Klop process on ordinals}, we
have successfully proved the following elegant result without any assumption:
\begin{theorem}{(Coarsest congruence contained in $\approx$, the final
    version)}
For any processes $p$ and $q$, \HOLinline{\HOLFreeVar{p} \HOLSymConst{\HOLTokenObsCongr} \HOLFreeVar{q}} if and only if
\HOLinline{\HOLFreeVar{p} \HOLSymConst{+} \HOLFreeVar{r} \HOLSymConst{\ensuremath{\approx}} \HOLFreeVar{q} \HOLSymConst{+} \HOLFreeVar{r}} for all processes $r$:
\begin{alltt}
\HOLTokenTurnstile{} \HOLSymConst{\HOLTokenForall{}}\HOLBoundVar{p} \HOLBoundVar{q}. \HOLBoundVar{p} \HOLSymConst{\HOLTokenObsCongr} \HOLBoundVar{q} \HOLSymConst{\HOLTokenEquiv{}} \HOLSymConst{\HOLTokenForall{}}\HOLBoundVar{r}. \HOLBoundVar{p} \HOLSymConst{+} \HOLBoundVar{r} \HOLSymConst{\ensuremath{\approx}} \HOLBoundVar{q} \HOLSymConst{+} \HOLBoundVar{r}\hfill[COARSEST_CONGR_FULL]
\end{alltt}
\end{theorem}

Going back to the congruence theory presented in previous section. Now
we can conclude that, the three relations (observation congruence,
weak bisimulation congruence, and the temporarily defined ``sum
equivalence'') coincide:
\begin{theorem}{(The equivalence of three relations)}
\begin{alltt}
\HOLTokenTurnstile{} \HOLConst{OBS_CONGR} \HOLSymConst{=} \HOLConst{SUM_EQUIV}
\HOLTokenTurnstile{} \HOLConst{OBS_CONGR} \HOLSymConst{=} \HOLConst{WEAK_CONGR}
\end{alltt}
\end{theorem}

\section{Conclusions}

In this project, we have done a further formalization of the process
algebra CCS in HOL4. Most results on strong equivalence, weak
equivalence and observation congruence were all formally proved.
A rather complete theory of congruence (for CCS) is also presented in
this project.

The project began with an old formalization of
CCS in Hol88 by Monica Nesi, then it's extended with formal proofs of
deep lemmas and theorems, including Hennessy Lemma, Deng Lemma, and
the ``coarsest congruence containing in weak equivalence'' theorem.
We believe these work have shown the possibility to use this project
as a research basis for discovering new theorems about CCS.

For the last theorem, we deeply investigated various versions of the
theorem and their proofs in original papers.
But for the most general
case, an infinite sum of CCS processes must be used during the proof,
and we had to add an axiom to assert the existence of infinite
sums. Without this axiom, the best result we could get is only for
finite-state CCS. The consistency of HOL logic after adding this axiom
is yet to be checked. On the other side, complete removing of this
axiom seems impossible in scope of higher order logic.

We have extensively used HOL's rich theories to simplify the
development efforts in this project, notable ones include:
\texttt{relationTheory} (for RTC) and \texttt{ordinalTheory} (for Klop
function defined on ordinals). Now we use HOL's
built-in co-inductive relation support to define strong and weak
equivalence, as the result many intermediate results were not needed
thus removed from the old scripts.

Some missing pieces include: the decision procedures for bisimilarity
checking (strong, weak and rooted week), HML and example models. For
these missing pieces, a further project is already in plan.

Thanks to Prof. Roberto Gorrieri, who taught CCS and LTS theory to the
author, and his supports on continuing this HOL-CCS project as exam
project of his course.

Thanks to Prof. Monica Nesi for finding and sending the old HOL88
proof scripts to the author.

Thanks to Prof. Andrea Asperti, who taught the interactive theorem proving
techniques to the author, although it's in another different theorem
prover (Matita).

Thanks to people from HOL community (Thomas Tuerk, Michael Norrish,
Ramana Kumar and many others) for resolving issues and doubts that the
author met when using HOL theorem prover.

The paper is written in \LaTeX and LNCS template, with theorems
generated automatically by HOL's \TeX
exporting module (\texttt{EmitTex}) from the proof scripts.

\bibliography{hol-ccs2}{}
\bibliographystyle{splncs}
\end{document}